\documentclass[a4paper,11pt]{article}
\usepackage{listings}
\usepackage{mathrsfs}
\usepackage[mathscr]{euscript}

\usepackage{vmargin}
\setmarginsrb{1in}{1in}{1in}{1in}{0pt}{0pt}{0pt}{6mm}
\usepackage{amsthm,amsmath,amssymb}
\usepackage{libertine}

\newtheorem{theorem}{Theorem}[section]
\newtheorem{lemma}[theorem]{Lemma}
\newtheorem{claim}[theorem]{Claim}
\newtheorem{corollary}[theorem]{Corollary}
\newtheorem{definition}[theorem]{Definition}
\newtheorem{observation}[theorem]{Observation}
\newtheorem{proposition}[theorem]{Proposition}
\newtheorem{remark}[theorem]{Remark}

\date{}

\usepackage{pgf}
\usepackage{paralist} 
\usepackage{graphicx}
\usepackage{boxedminipage}
\usepackage{boxedminipage}
\usepackage{xcolor}
\usepackage[linesnumbered,boxed,algosection]{algorithm2e}
\usepackage{framed}

\usepackage{amsmath, amssymb, latexsym}
\usepackage{enumerate}

\usepackage{float}
\usepackage{xspace}

\usepackage{caption}
\usepackage{subcaption}

\usepackage{todonotes}

\usepackage{booktabs}

 \usepackage[ pdftex, plainpages = false, pdfpagelabels, 
                 bookmarks=false,
                 bookmarksopen = true,
                 bookmarksnumbered = true,
                 breaklinks = true,
                 linktocpage,
                 pagebackref,
                 colorlinks = true,  
                 linkcolor = blue,
                 urlcolor  = blue,
                 citecolor = red,
                 anchorcolor = green,
                 hyperindex = true,
                 hyperfigures
                 ]{hyperref}

\usepackage{thm-restate}

\usepackage{thmtools}
\usepackage{cleveref}

\usepackage{mathtools}

\usepackage{xspace}
\newcommand{\SESCAD}{\textsc{SESCAD}\xspace}
\newcommand{\ESCAD}{\textsc{ESCAD}\xspace}
\newcommand{\XP}{\textsf{XP}\xspace}
\newcommand{\type}{\mathsf{Types}}

\newcommand{\clique}{\ell}

\newcommand{\PSI}{\normalfont\textsc{Partitioned Subgraph Isomorphism}\xspace}
\newcommand{\PSIshort}{\normalfont\textsc{PSI}\xspace}

\newcommand{\ETH} {{\sf ETH}\xspace}

\newcommand{\fpt} {{\sf FPT}\xspace}
\newcommand{\nph} {{\sf NP}-hard\xspace}

\newcommand{\woh} {{\sf W}$[1]$-hard\xspace}

\newcommand{\tw}{{\mathtt{tw}}}

\newcommand{\OO}{{\mathcal O}}
\newcommand{\CC}{{\mathcal C}}
\newcommand{\TT}{{\mathcal T}}

\newcommand{\PP}{{\mathcal P}}

\newenvironment{claimproof}[1][\proofname]
  { \proof[#1]  }
  { \endproof }

\usepackage{tabularx}
\usepackage{tcolorbox}
\newcommand{\defprob}[3]{
    \begin{tcolorbox}
        \begin{minipage}{0.99\textwidth}
            \begin{tabular*}{\textwidth}{@{\extracolsep{\fill}}ll} #1   \\ \end{tabular*}
            {\bf{Input:}} #2  \\
            {\bf{Question:}} #3
        \end{minipage}
    \end{tcolorbox}
}

\begin{document}

\title{On the Parameterized Complexity of Eulerian Strong Component Arc Deletion \footnote{An extended abstract of this paper appears in the proceedings of IPEC 2024 \cite{DBLP:conf/iwpec/BlazejJ0S24}.}}

 \author{
Václav Blažej \thanks{University of Warwick, \textrm{vaclav.blazej@warwick.ac.uk.}}
 \and Satyabrata Jana \thanks{University of Warwick, \textrm{satyamtma@gmail.com.}}
 \and M. S. Ramanujan \thanks{University of Warwick, \textrm{r.maadapuzhi-sridharan@warwick.ac.uk.}}
 \and Peter Strulo \thanks{University of Warwick, \textrm{Peter.Strulo@warwick.ac.uk.}}
}

 \maketitle

\thispagestyle{empty}
\begin{abstract}

In this paper, we study the Eulerian Strong Component Arc Deletion problem, where the input is a directed multigraph and the goal is to delete the minimum number of arcs to ensure every strongly connected component of the resulting digraph is Eulerian.  

This problem is a natural extension of the Directed Feedback Arc Set problem and is also known to be motivated by certain scenarios arising in the study of housing markets. The complexity of the problem, when parameterized by solution size (i.e., size of the deletion set), has remained unresolved and has been highlighted in several papers.
In this work, we answer this question by ruling out (subject to the usual complexity assumptions) a fixed-parameter algorithm (\fpt algorithm) for this parameter and conduct a broad analysis of the problem with respect to other natural parameterizations. We prove both positive and negative results. Among these, we demonstrate that the problem is also hard (W[1]-hard or even para-NP-hard) when parameterized by either treewidth or maximum degree alone. Complementing our lower bounds, we establish that the problem is in XP when parameterized by treewidth and \fpt when parameterized either by both treewidth and maximum degree or by both treewidth and solution size. We show that on simple digraphs, these algorithms have near-optimal asymptotic dependence on the treewidth assuming the Exponential Time Hypothesis.
\end{abstract}

\newpage
\section{Introduction}\label{sec:intro}

In the Eulerian Strong Component Arc Deletion (\ESCAD) problem, the input is a directed graph (digraph)
and a number $k$ and the goal is to delete at most $k$ arcs to ensure every strongly connected component of the resulting digraph is Eulerian.  
This problem was first introduced by Cechlárová and Schlotter \cite{DBLP:conf/iwpec/CechlarovaS10} to model problems arising in the study of housing markets and they left the existence of an {\fpt} algorithm
for \ESCAD as an open question.

The \ESCAD problem extends the well-studied Directed Feedback Arc Set (DFAS) problem. In DFAS, the goal is to delete the minimum number of arcs to make the digraph acyclic. The natural 
extension of DFAS to \ESCAD introduces additional complexity as we aim not to prevent cycles, but aim to balance in-degrees and out-degrees within each strongly connected component. As a result, the balance requirement complicates the problem significantly and the ensuing algorithmic challenges have been noted in multiple papers~\cite{DBLP:conf/iwpec/CechlarovaS10,DBLP:journals/algorithmica/CyganMPPS14,DBLP:journals/disopt/GokeMM22}.

Crowston et al.~\cite{DBLP:journals/ipl/CrowstonGJY12} made partial progress on the problem by showing that \ESCAD is fixed-parameter tractable (\fpt) on tournaments and also gave a polynomial kernelization. However, the broader question of fixed-parameter tractability of {\ESCAD} on general digraphs has remained unresolved.

\smallskip
\noindent
\textbf{Our contributions.}
Our first main result rules out the existence of an {\fpt} algorithm for \ESCAD under the solution-size parameterization, subject to the standard complexity-theoretic assumption that {\sf W}$[1]\neq $  {\fpt} .  Moreover, we show that assuming  
the Exponential Time Hypothesis (\ETH),  the trivial $n^{\OO(k)}$ algorithm where we simply guess and verify the solution is asymptotically near-optimal. It is well known that if \ETH is true, then {\sf W}$[1]\neq $  {\fpt} \cite{DBLP:series/txcs/DowneyF13}. 

\begin{restatable}{theorem}{ESCADSOLUTION}\label{theo:escadsolution}
    \ESCAD is \woh parameterized by the solution size $k$. Moreover, there is no algorithm solving \ESCAD in $f(k)\cdot n^{o(k/ \log k)}$ time for some function $f$, where $n$ is the input length, unless the Exponential Time Hypothesis fails.
\end{restatable}

The above negative result explains, in some sense, the algorithmic challenges encountered in previous attempts at showing tractability and shifts the focus toward alternative parameterizations. In particular, we explore various structural parameterizations of the  {\em undirected graph underlying the input digraph}.
We first show that even a strong parameterization such as the vertex cover number is unlikely to lead to a tractable outcome and also obtain an {\ETH}-based lower bound.

\begin{restatable}{theorem}{ESCADWHARDVC}\label{thm:escad_whard_vc}\label{thm:escad_eth_lb}
    \ESCAD is \woh parameterized by the vertex cover number $\mathrm{vc}$ of the graph. Moreover, there is no algorithm solving ESCAD in $f(\mathrm{vc})\cdot n^{o(\mathrm{vc}/ \log \mathrm{vc})}$ time for some function $f$,
    where    $n$ is the input length, unless the Exponential Time Hypothesis fails.
\end{restatable}

To add to the hardness results above, we also analyze the parameterized complexity of the problem parameterized by the maximum degree of the input digraph and show that even for constant values of the parameter, the problem remains {\nph}.

\begin{restatable}{theorem}{ESCADNPHDEGREE}\label{theo:nph}
\ESCAD is \nph~in  digraphs where each vertex has $(\mathrm{in}, \mathrm{out})$ degrees in $\{(1,6), (6,1)\}$.
\end{restatable}

We complement these negative results by showing that \ESCAD is {\fpt} parameterized by the treewidth of the graph and solution size as well as by the treewidth and maximum degree. Furthermore, we give an {\XP} algorithm parameterized by treewidth alone. All three results are obtained through the same algorithm. The formal statement follows.

\newcommand{\bigoh}[0]{\OO}

\begin{restatable}{theorem}{ESCADTWXP}\label{thm:escad_tw_xp}
    An \ESCAD instance $\mathcal I = (G,k)$ can be solved by an algorithm in time $\alpha^{\OO(\tw^2)} \cdot n^{\OO(1)}$ where $\tw$ is the treewidth of $G$, $\Delta$ is the maximum degree of $G$, and $\alpha = \min(k,\Delta)$. Additionally, if $G$ excludes multiarcs, then this algorithm runs in time $2^{\OO(\tw^2)} \cdot \alpha^{\OO(\tw)} \cdot n^{\OO(1)}$. 
\end{restatable}

In the above statement, notice that $\alpha$ is upper bounded by the number of arcs in the digraph and so, implies an {\XP} algorithm parameterized by the treewidth with running time $n^{\OO(\tw^{2})}$ in general and an algorithm with running time $2^{\OO(\tw^2)} \cdot n^{\OO(\tw)}$ on simple digraphs (digraphs without multiarcs or loops).

The distinction between general instances of {\ESCAD} and instances on simple digraphs appears to be fundamental. 
In fact, multiarcs are crucially used in the proof of \Cref{thm:escad_whard_vc}, raising the question of adapting this reduction to {\em simple digraphs}  in order to obtain a similar hardness result parameterized by vertex cover number. However, we show that this is not possible by giving an {\fpt} algorithm for the problem on simple digraphs parameterized by the vertex integrity of the input graph. Recall that a digraph has \emph{vertex integrity} $k$ if there exists a set of vertices of size $q\leq k$ which when removed, results in a digraph where each weakly connected component has size at most $k-q$. Vertex integrity is a parameter lower bounding vertex cover number and has gained popularity in recent years as a way to obtain {\fpt} algorithms for problems that are known to be {\woh} parameterized by treedepth -- one example being {\ESCAD} on simple graphs as we show in this paper (see \Cref{thm:combinedSESCAD} below). 

\begin{restatable}{theorem}{sescadVIFPT}\label{thm:sescadVIFPT}
    ESCAD on simple digraphs is \fpt parameterized by the vertex integrity of the graph.
\end{restatable}

As a consequence of this result, we infer an \fpt  algorithm for \ESCAD on simple digraphs parameterized by the vertex cover number, highlighting the difference in the behaviour of the \ESCAD problem on directed graphs that permit multiarcs versus simple digraphs. 
On the other hand, we show that even on simple digraphs this positive result does not extend much further to well-studied width measures such as treewidth (or the even larger parameter treedepth), by obtaining the following consequence of \Cref{thm:escad_whard_vc,thm:escad_eth_lb}.

\begin{theorem}\label{thm:combinedSESCAD}
    \ESCAD even on simple digraphs is \woh parameterized by $k$,  where $k$ is the size of the smallest vertex set that must be deleted from the input digraph to obtain a disjoint union of directed stars. Moreover, assuming \ETH, there is no algorithm solving it in $f(k)n^{(k/ \log k)}$ time for some function $f$, where $n$ is the input length. 
\end{theorem}

In the above statement, a {\em directed star} is just a digraph whose underlying undirected graph is a star. 
Furthermore, notice the running time of our algorithm for {\ESCAD} on simple digraphs in Theorem \ref{thm:escad_tw_xp} asymptotically almost matches our {\ETH}-based lower bound in \Cref{thm:combinedSESCAD}. This is because  the size of the smallest vertex set that must be deleted from the input digraph to obtain a disjoint union of directed stars is at least the treewidth minus one.

\medskip
\noindent
\textbf{Related Work.}
The vertex-deletion variant of \ESCAD is known to be {\woh} parameterized by the solution size, as shown by G\"{o}ke et al.~\cite{DBLP:journals/disopt/GokeMM22}, who identify \ESCAD as an open problem and note that gaining more insights into its complexity was a key motivation for their study.  Cygan et al.~\cite{DBLP:journals/algorithmica/CyganMPPS14} gave the first \fpt algorithm for edge (arc) deletion to Eulerian graphs (respectively, digraphs). Here, the aim is to make the whole graph Eulerian whereas the focus in \ESCAD is on each strongly connected component. Cygan et al. also explicitly highlight \ESCAD as an open problem and a motivation for their work. Goyal et al.~\cite{DBLP:journals/jcss/GoyalMPPS18} later improved the algorithm of Cygan et al. by giving algorithms achieving a single-exponential dependence on $k$.

\section{Preliminaries}\label{sec:prelims}

In this paper, the arc set of a digraph is a multiset, i.e., we allow multiarcs. Moreover, we treat multiarcs between the same ordered pairs of vertices as distinct arcs in the input representation of all digraphs. Consequently, the number of arcs in the input is upper bounded by the length of the input. We exclude loops as they play no non-trivial role in instances of this problem.

For a digraph $G$, we denote its vertices by $V(G)$, arcs by $E(G)$, the subgraph induced by $S \subseteq V(G)$ as $G[S]$,  a subgraph with $S$ removed as $G - S = G[V(G) \setminus S]$, and a subgraph with subset of arcs $F \subseteq E(G)$ removed as $G-F= (V(G), E(G) \setminus F)$. Unless otherwise specified, all paths and cycles we consider are simple. 
For a vertex $v$ and digraph $G$, let $\deg^-_G(v)$ denote its in-degree, $\deg^+_G(v)$ be its out-degree, and $\deg^+_G(v)-\deg^-_G(v)$ is called its {\em imbalance} in $G$.
If the {imbalance} of $v$ is $0$ then $v$ is said to be \emph{balanced} (in $G$).
A digraph is called \emph{balanced} if all its vertices are balanced.
The maximum degree of a digraph $G$ is the maximum value of $\deg^+_G(v)+\deg^-_G(v)$ taken over every vertex $v$ in the graph.
We denote the empty graph $(\emptyset, \emptyset)$ by simply $\emptyset$, similarly we use the same symbol for the unique function of type $\emptyset \rightarrow \emptyset$. It will be clear from the context which of these objects is meant.

A vertex $v$ is \emph{reachable} from $u$ if there exists a directed path from $u$ to $v$ in $G$. 
A \emph{strongly connected component} of $G$ is a maximal set of vertices where all vertices are mutually reachable from each other.
Let the \emph{strong subgraph} of $G$,  denoted by $\mathrm{strong}(G)$, be the subgraph of $G$ obtained by removing all arcs that have endpoints in different strongly connected components.
The \ESCAD problem can now be formulated as ``Is there a set $S \subseteq E(G)$ of size $|S| \le k$ such that $\mathrm{strong}(G - S)$ is balanced?''
We call an arc $e \in E(G)$ \emph{active} in $G$ if $e \in E(\mathrm{strong}(G))$ and \emph{inactive} in $G$ otherwise.

Given a graph $G$, a subset of vertices $S\subseteq V(G)$ is a \emph{vertex cover} if $G-S$ is an independent set. We say that $G$ has \emph{vertex cover number} $k$ if there exists a vertex cover $S$ where $|S| \le k$.  A star is an undirected graph isomorphic to $K_1$ or $K_{1,t}$ for some $t\geq 1$. 

A \emph{tree decomposition} of an undirected graph $G$ is a pair $(T, \{X_t\}_{t \in V(T)})$ where $T$ is a tree and $X_t \subseteq V(G)$ such that (i) for all edges $uv \in E(G)$ there exists a node $t \in V(T)$ such that $\{u, v\} \subseteq X_t$ and (ii) for all $v \in V(G)$ the subgraph induced by $\{t \in V(T) : v \in X_t\}$ is a non-empty tree. The \emph{width} of a tree decomposition is $\max_{t \in V(T)}|X_t| - 1$. The \emph{treewidth} of $G$ is the minimum width of a tree decomposition of $G$.

Let $(T, \{X_t\}_{t \in V(T)})$ be a tree decomposition of $G$. We refer to every node of $T$ with degree one as a \emph{leaf node} except one which is chosen as the root, $r$. A tree decomposition $(T, \{X_t\}_{t \in V(T)})$ is a \emph{nice tree decomposition with introduce edge nodes} if all of the following conditions are satisfied:
\begin{enumerate}
    \item $X_r = \emptyset$ and $X_\ell = \emptyset$ for all leaf nodes $\ell$.
    \item Every non-leaf node of $T$ is one of the following types:
    \begin{itemize}
        \item \textbf{Introduce vertex node:} a node $t$ with exactly one child $t'$ such that $X_t = X_{t'} \cup \{v\}$ for some vertex $v \notin X_{t'}$; we say that the vertex $v$ is \emph{introduced} at $t$.
        \item \textbf{Introduce edge node:} a node $t$, labeled with an edge $uv$ where $u,v \in X_t$ and with exactly one child $t'$ such that $X_t = X_{t'}$; we say that the edge $uv$ is \emph{introduced} at $t$.
        \item \textbf{Forget node:} a node $t$ with exactly one child $t'$ such that $X_t = X_{t'} \setminus \{v\}$ for some vertex $v \in X_{t'}$; ; we say that the vertex $v$ is \emph{forgotten} at $t$.
        \item \textbf{Join node:} a node $t$ with exactly two children $t_1, t_2$ such that $X_t = X_{t_1} = X_{t_2}$.
    \end{itemize}
    \item Every edge appears as the label of exactly one introduce edge node.
\end{enumerate}
\section{Our Results for \ESCAD}

In the following four subsections we describe three hardness results and tractability results on bounded treewidth graphs for \ESCAD.
In \Cref{sec:escad_tw_tractability} we show that the problem is \XP by treewidth and \fpt in two cases -- when parameterized by the combined parameter treewidth plus maximum degree, and when parameterized by treewidth plus solution size.
The hardness results show that dropping any of these parameters leads to a case that is unlikely to be \fpt.
More precisely, we show that parameterized by solution size it is \woh (in \Cref{sec:escad_w1_solution}) as is the case when parameterized by vertex cover number (\Cref{sec:escad_w1_hard_vc}), and it is para-\nph~when parameterized by the maximum degree (\Cref{sec:escad_np_hard_maxdegree}).

\subsection{W[1]-hardness of ESCAD Parameterized by Solution Size}\label{sec:escad_w1_solution}

In this section, we show that \ESCAD is \woh when parameterized by solution size.
Our reduction is from \PSI (\PSIshort).
Intuitively, in \PSIshort we are given vertex-colored graphs $G'$ and $H'$ and aim to find $H'$ as a subgraph of $G'$ while respecting the vertex colors.
Formally, the input is a pair of graphs $G'$ and $H'$ with $|V(H')| \le |V(G')|$ and a mapping $\psi\colon V(G') \to V(H')$.
Let $\clique = |V(H')|$ and let us denote vertices of $H'$ by integers, also called \emph{colors}, i.e., $V(H') = \{1,\dots,\clique\}$.
The goal is to decide whether there exists an injective mapping $\phi\colon V(H') \to V(G')$ such that $\{\phi(u),\phi(v)\} \in E(G')$ for each $\{u,v\} \in E(H')$ and $\psi \circ \phi$ is the identity.

Notice that the classic {\sc Multicolored Clique} problem is a special case of {\PSIshort}, where $H'$ is a complete graph and since the former is \woh \cite{DBLP:books/sp/CyganFKLMPPS15}, it follows that so is \PSIshort.
Moreover, assuming {\ETH}, there is no $f(|E(H')|)\cdot n^{o(|E(H')|/\log |E(H')|)}$-algorithm for {\PSIshort} \cite[Corollary 6.3]{Marx10}.

For each $i \in [\clique]$ let $V_i = \psi^{-1}(i)$.
We call $V_i$ a \emph{color class} and for a vertex $v$ in $G'$, we say $v$ has color $i$ if $v \in V_i$.
By $N_{H'}(i)$ we denote the neighborhood of $i \in [\clique]$ in $H'$.
We assume without loss of generality that in the \PSI instance we reduce from, each color class $V_i$ forms an independent set (edges in the same color class can be removed) and for each vertex $v \in V_i$ and each $j \in N_{H'}(i)$ there exists a $w \in V_j$ that is adjacent to $v$ (if not, then such a vertex cannot participate in the solution and can be removed).

We start with descriptions of two auxiliary gadgets: the {\em imbalance gadget} and the \emph{path gadget}.

\medskip
\noindent
\textbf{Imbalance Gadget.}
Let $u, v$ be a pair of vertices, and $b,c$ be two positive integers.
We construct a gadget $I_{u,v}$ connecting the vertex $u$ to $v$ by a path with vertices $u,w_1,\dots,w_b,v$ where $w_i$'s are $b$ new vertices (we call them intermediate vertices in this gadget); let $w_0=u$ and $w_{b+1}=v$.
For every $i \in \{0,\dots,b\}$ the path contains $b+1+c$ forward arcs $(w_i,w_{i+1})$ and $b+1$ backward arcs $(w_{i+1},w_i)$, see \Cref{fig:imbalance_gadget_2} for an illustration.
Observe that in the gadget $I_{u,v}$, vertex $u$ has imbalance $c$ and vertex $v$ has imbalance $-c$, whereas intermediate vertices have imbalance zero.
We refer to this gadget $I_{u,v}$ as a \emph{$(b,c)$-imbalance gadget}.

\begin{figure}[ht]
    \centering
    \hfill
    \begin{subfigure}[b]{0.45\textwidth}
         \centering
         \includegraphics[page=1,scale=1]{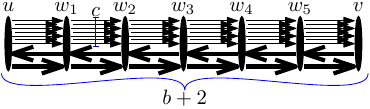}
         \caption{$(b,c)$-imbalance gadget}%
         \label{fig:imbalance_gadget_2}
    \end{subfigure}
    \hfill
    \begin{subfigure}[b]{0.45\textwidth}
         \centering
         \includegraphics[page=2,scale=1]{w1_solution}
         \caption{$(b,c)$-path gadget}%
         \label{fig:path_gadget}
    \end{subfigure}
    \hfill
    \caption{Black ellipses are vertices, each thick arc represents $(b+1)$ copies of an arc; $(b+2)$ is the number of vertices in the gadgets.}%
    \label{fig:w1_solution}
\end{figure}

\medskip
\noindent
\textbf{Path Gadget.} Let $u, v$ be a pair of vertices, and $b,c$ be two positive integers.
We construct a gadget $P_{u,v}$ connecting the vertex $u$ to $v$ by a path with vertices $u,w_1,\dots,w_b,v$ where $w_i$'s are $b$ new intermediate vertices; let $w_0=0$ and $w_{b+1}=v$.
For every $i \in \{0,\dots,b\}$ the path contains $c$ forward arcs $(w_i,w_{i+1})$.
See \Cref{fig:path_gadget} for an illustration.
Notice that, unlike the imbalance gadget, we do not add backward arcs.
Observe that in the gadget $P_{u,v}$, vertex $u$ has imbalance $c$ and vertex $v$ has imbalance $-c$, whereas intermediate vertices have imbalance zero.
We refer to this gadget $P_{u,v}$ as a \emph{$(b,c)$-path gadget}.

We use the following properties of the gadgets $I_{uv}$ and $P_{uv}$ to reason about the correctness of our construction.
We say a (path or imbalance) gadget $Q$ is \emph{present} in $G$ if $Q$ is an induced subgraph of $G$ and its intermediate vertices have no neighbors in $V(G) \setminus V(Q)$.

\begin{lemma}\label{lem:imbalance_gadget}
    Let $(G,b)$ be a yes-instance of \ESCAD and $S$ be a solution.
    Assume that for a pair of vertices $u,v$ in $G$, there is a $(b,c)$-imbalance gadget $I_{uv}$ present in $G$.
    If $S$ is an inclusion-wise minimal solution then $S$ contains no arc of $I_{uv}$.
\end{lemma}

\begin{proof}
    In the subgraph $I_{uv}$, there are $b+c+1$ arc-disjoint paths from $u$ to $v$ and $b+1$ arc-disjoint paths from $v$ to $u$.
    Both quantities are more than the budget $b$, so at least one $u$-$v$ path and one $v$-$u$ path remain disjoint from $S$. Since each of these paths contains all the vertices in $I_{uv}$, it follows that all vertices in $I_{uv}$ must be contained in the same strongly connected component of $G-S$.
    For each $i\in [b+1]$, let $F_i$ be all the arcs between $w_{i-1}$ and $w_i$ in either direction.
    If $S$ contains no arc of $I_{uv}$, then we are done.
    Otherwise, as $|S| \leq b$ and $S \cap E(I_{uv}) \neq \emptyset$, there is $i \in [b]$ such that either $S \cap F_i \neq \emptyset$, $S \cap F_{i+1} = \emptyset$ or $S \cap F_i = \emptyset$, $S \cap F_{i+1} \neq \emptyset$.
    Assume that $S \cap F_i \neq \emptyset$ and $S \cap F_{i+1} = \emptyset$.
    The argument in the other case is analogous.
    To ensure that the imbalance of $w_i$ is zero in $\mathrm{strong}(G-S)$, the solution $S$ must contain the same number of out-arcs and in-arcs of $w_i$ from the set $F_i$ (due to the fact that $w_i$ and $w_{i-1}$ are contained in the same strongly connected component of $G-S$ and $F_{i+1} \cap S=\emptyset$).
    Now, consider the set $S' = S \setminus (S \cap F_i)$.
    As the vertices $w_{i-1}$ and $w_{i}$ are in the same strongly connected component of $G-S$, and $S'$ is a subset of $S$, they must be in the same strongly connected component of $G-S'$.
    We saw that $S \cap F_i$ has the same number of forward and backward arcs, therefore, in $G-S'$ the only vertices with changed degrees, $w_i$ and $w_{i-1}$, remain balanced.
    So, $S' \subset S$ is also a solution, contradicting the choice of inclusion-wise minimal $S$.
\end{proof}

\begin{lemma}\label{obs:path_gadget}
    Let $(G,b)$ be a yes-instance of \ESCAD and $S$ be an inclusion-wise minimal solution for this instance.
    For a pair of vertices $u,v$ in $G$, assume that for an arbitrary $c$ there is a $(b,c)$-path gadget $P_{uv}$ present in $G$ and there are more than $b$ arc-disjoint paths from $v$ to $u$ (in $G - E(P_{uv})$).
    If $S$ contains an arc from $P_{uv}$, then there exists $i \in \{0,\dots,b\}$ such that $S \cap P_{uv}$ is equal to all $c$ arcs $(w_i,w_{i+1})$ of $P_{uv}$.
\end{lemma}

\begin{proof}
    The argument is similar to the proof of Lemma \ref{lem:imbalance_gadget}.
    For each $i\in [b+1]$, let $F_i$ be all the arcs between $w_{i-1}$ and $w_i$.
    If $S$ constains no arc of $P_{uv}$, then we are done.
    Otherwise, in case $F_i \subseteq S$ for some $i \in [b+1]$ by minimality of $S$ we have $F_i = S \cap E(P_{uv})$, as we wanted to show.
    In the other case, there is no $i \in [b+1]$ such that $F_i \subseteq S$.
    Hence, there is a path $P$ from $u$ to $v$ in $G-S$ such that all the vertices of $P$ belong to $P_{uv}$.
    As $|S| \leq b$ and $S \cap E(P_{uv}) \neq \emptyset$, there is $i \in [b]$ such that either $S \cap F_i \neq \emptyset$, $S \cap F_{i+1} = \emptyset$ or $S \cap F_i = \emptyset$, $S \cap F_{i+1} \neq \emptyset$.
    Assume that $S \cap F_i \neq \emptyset$ and $S \cap F_{i+1} = \emptyset$.
    The argument in the other case is analogous.
    Since there is a path from $u$ to $v$ in $G-S$ and there are more than $b$ arc-disjoint paths from $v$ to $u$ in $G-E(P_{uv})$, it follows that $u$ and $v$ are in the same strongly connected component of $G-S$.
    Moreover, the path $P$ contains the vertices $u,w_{i-1},w_i,w_{i+i},v$ so vertices $w_{i-1},w_i,w_{i+1}$ are in the same strongly connected component of $G-S$.
    As $S \cap F_i \neq \emptyset$ and $S \cap F_{i+1} = \emptyset$ the vertex $w_i$ is not balanced in $\mathrm{strong}(G-S)$, a contradiction.
\end{proof}

We are now ready to give our W[1]-hardness reduction.

\ESCADSOLUTION*

\begin{proof}
    Consider an instance $\mathcal{I'}=(G',H',\psi)$ of \PSI with $n$ vertices.
    Recall $V_i=\psi^{-1}(i)$ and our assumption that each color class induces an independent set, and for each color $i$ every vertex $u \in V_i$ neighbors at least one vertex of color $j$ if $ij \in E(H')$. In addition, we assume that for every pair of color classes $i,j$, if $ij\notin E(H')$, then there is no edge with one endpoint in $V_i$ and the other in $V_j$. 
    In polynomial time, we construct an \ESCAD instance $\mathcal I=(G,k)$ in the following way (see \Cref{fig:w1_hardness_sol_reduction} for an overview).

    \begin{itemize}
        \item We set $k = 4 \cdot |E(H')|$.

        \item We construct $V(G)$ as follows:
        \begin{enumerate}
            \item We add a vertex $s$.
            \item For each color $i \in [\clique]$, we add a pair of vertices $s_i$ and $d_i$.
            \item For each vertex $u$ in $V(G')$, we add a vertex $x_u$.
            \item For each edge $uv$ in $E(G')$, we add a vertex $z_{uv}$.
        \end{enumerate}

        \item We construct $E(G)$ as follows. We introduce four sets of arcs $E_1$, $E_2$, $E_3$, and $E_4$ that together comprise the set $E(G)$. For each color $i \in [\clique]$, let $r_i \coloneqq |V_i| \cdot |N_{H'}(i)|$, $c_i \coloneqq |\{uv \mid  uv\in E(G'), u\in V_i\}|-r_i$. Notice that $c_i \geq 0$ due to the assumptions.
        \begin{enumerate}
            \item For each $i \in [\clique]$, we add a $(k, r_i-|N_{H'}(i)|)$-imbalance gadget $I_{s,d_i}$ and a $(k, c_i)$-imbalance gadget $I_{s,s_i}$ to $E_1$.

            \item For each $i \in [\clique]$, for each vertex $u \in V_i$ we add a $(k, |N_{G'}(u)|-|N_{H'}(i)|)$-imbalance gadget $I_{s_i,x_u}$ and a $(k, |N_{H'}(i)|)$-path gadget $P_{d_i,x_u}$ to $E_2$.

            \item For every edge $uv \in E(G')$, we add a pair of arcs $(x_u,z_{uv})$ and $(x_v,z_{uv})$ to $E_3$. 

            \item For every edge $uv \in E(G')$, we add two copies of the arc $(z_{uv},s)$ to $E_4$.
        \end{enumerate}
    \end{itemize}

    \begin{figure}[ht]
        \centering
        \includegraphics[page=3]{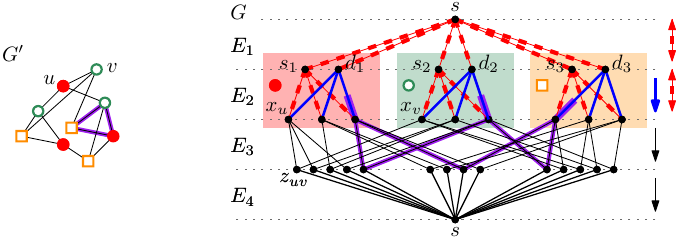}
        \caption{
            Overview of the reduction from $G'$ to $G$;
            four sets of arcs are depicted from top to bottom.
            $E_1$ contains imbalance gadgets,
            $E_2$ is a mixture of imbalance and path gadgets,
            $E_3$ contains arcs,
            and $E_4$ has double arcs to $s$.
            Marked purple arcs correspond to a solution in $G'$ and its respective solution in $G$.
            The thee colored backgrounds in $G$ signify parts of the construction tied to the three color classes.
            Direction of arcs in $G$ is shown by the arrows on the right.
            The picture of $G$ wraps around to the top as the vertex $s$ drawn on the bottom is the same as the one drawn on the top.
        }%
        \label{fig:w1_hardness_sol_reduction}
    \end{figure}

    It is easy to see that the construction can be performed in time polynomial in $|V(G')|$. Moreover, the reduction transforms the parameter linearly and the instance size polynomially. This implies the {\ETH}-based lower bound. 
    Hence, it only remains to prove the correctness of our reduction.
    
    As each vertex of $G$ lies on a cycle that goes through $s$, it follows that $G$ is strongly connected. First, we argue about the imbalances of the vertices in $G$.

    \begin{claim}\label{claim:imbalance}
        The only vertices with non-zero imbalance in $G$ are those in the set $\{s\} \cup \{ d_i \mid  i\in [\clique] \}$. Furthermore, the imbalance of the vertex $s$ is $-2|E(H')|$ and the imbalance of $d_i$ for each $i \in [\clique]$ is $|N_{H'}(i)|$.
    \end{claim}

    \begin{claimproof}
        There are six types of vertices in $G$ --
        (1) the vertex $s$,
        (2) vertices $s_i$ for $i\in [\clique]$,
        (3) vertices $d_i$ for $i\in [\clique]$,
        (4) vertices $x_u$ for $u\in V(G')$,
        (5) vertices $z_{uv}$ for $uv\in E(G')$,
        and (6) the intermediate vertices (in the imbalance gadgets and path gadgets).
        Below, we examine their imbalance one by one in the given order.

        \begin{enumerate}[\bf (1)]
            \item Due to the imbalance gadgets in $E_1$, for $i \in [\clique]$, the vertex $s$ incurs imbalance $r_i-|N_{H'}(i)|$ from $I_{s,d_i}$ and $c_i$ imbalance from $I_{s,s_i}$. Summing up and simplifying, in total the vertex $s$ has imbalance $2|E(G')|-2|E(H')|$ in $G[E_1]$. In $E_4$ the vertex $s$ has $2|E(G')|$ incoming arcs. Hence, the imbalance of $s$ in $G$ is $-2|E(H')|$.

            \item Fix $i \in [\clique]$. Due to the imbalance gadgets in $E_1$ the vertex $s_i$ gets imbalance $-c_i$ from $I_{s,s_i}$. Due to $I_{s_i,x_u}$ in $E_2$ the vertex $s_i$ gets imbalance $|N_{G'}(u)|-|N_{H'}(i)|$ for every $u \in V_i$, which in total equals $\sum_{u \in V_i} \big(|N_{G'}(u)|-|N_{H'}(i)|\big) = |\{uv \mid  uv\in E(G'), u\in V_i\}|-|V_i| \cdot |N_{H'}(i)| = c_i$. So $s_i$ is balanced in $G$.

            \item Fix $i \in [\clique]$. Due to $I_{s,d_i}$ in $E_1$ the vertex $d_i$ gets imbalance $-r_i+|N_{H'}(i)|$. Due to $P_{d_i,x_u}$ in $E_2$, the vertex $d_i$ gets imbalance $|N_{H'}(i)|$ for every $u \in V_i$. In total, this is $|V_i| \cdot |N_{H'}(i)| = r_i$. The total imbalance of $d_i$ is therefore, $|N_{H'}(i)|$.

            \item Fix $i \in [\clique]$ and $u \in V_i$. From $E_2$ vertex $x_u$ has imbalance $-|N_{G'}(u)|+|N_{H'}(i)|$ due to $I_{s_i,x_u}$ and imbalance $-|N_{H'}(i)|$ due to $P_{d_i,x_u}$ which sums to $-|N_{G'}(u)|$. In $E_3$ the vertex $x_u$ has $|N_{G'}(u)|$ outgoing arcs, so $x_u$ is balanced in $G$.

            \item For each $uv \in E(G)$, the vertex $z_{uv}$ is incident to exactly two incoming and two outgoing arcs so it is balanced in $G$.

            \item The remaining vertices are the intermediate vertices of imbalance and path gadgets. Those are balanced by construction.
                \qedhere
        \end{enumerate}
    \end{claimproof}

    This shows that there are only $\clique+1$ vertices with non-zero imbalance in $G$. The imbalances of the $d_i$ vertices will make us ``choose'' vertices and edges that represent $H'$ in $G'$ as we will see below.

    We now show correctness of the reduction.
    In the forward direction, assume that $(G', H', \psi)$ is a yes-instance, so there is a solution mapping $\phi\colon V(H') \to V(G')$.
    Let $v_i=\phi(i)$ for each $i \in [\clique]$ and let $K=(\{\phi(i) \mid  i \in V(H')\}, \{\phi(i)\phi(j) \mid ij \in E(H')\})$ be the \PSIshort solution subgraph in $G'$.
    Note that $v_i \in V_i$.
    We now construct a solution $S$ of $(G, k)$.
    For each edge $v_iv_j \in E(H')$ we add the arcs $(x_{v_j},z_{v_iv_j})$ and $(x_{v_i},z_{v_iv_j})$ to $S$.
    There are $2 \cdot |E(H')|$ many such arcs.
    Now for each $i \in [\clique]$ we add all the incoming arcs of $x_{v_i}$ along the path gadget $P_{d_i x_{v_i}}$ to $S$.
    As for each $i \in [\clique]$, the number of such arcs is $|N_{H'}(i)|$ we have $|S| = 2 \cdot |E(H')| + \sum_{i \in [\clique]}|N_{H'}(i)| = 4 \cdot |E(H')| = k$.
    Now, we show that each strongly connected component of $G-S$ is Eulerian.
    For an example of $S$, refer to \Cref{fig:w1_hardness_sol_reduction} (purple arcs).

    We consider the strongly connected components of $G-S$ and we will show that each of them is Eulerian. We first define:
    \[
        Z=\{z_{uw} \mid uw \in E(K)\} \cup \big(\bigcup_{i \in [\clique]}(V(P_{d_i, x_{v_i}}) \setminus \{d_i, x_{v_i}\})\big)
    \]

    \begin{claim}\label{claim:connected_component}
        One strongly connected component of $G-S$ consists of all the vertices except $Z$ and all other strongly connected components of $G-S$ are singleton -- one for each vertex in $Z$.
    \end{claim}
    \begin{claimproof}
        We have added to $S$ all incoming arcs of the vertices in $\{z_{uw} \mid uw \in E(K)\}$, so each forms a singleton strongly connected component of $G-S$.
        For the remaining vertices of $Z$, fix $i \in [\clique]$ and observe that $x_{v_i}$ is a sink vertex in the path gadget $P_{d_i x_{v_i}}$.
        Hence, every cycle of $G$ that contains a vertex from $V(P_{d_i, x_{v_i}}) \setminus \{d_i, x_{v_i}\}$ has to include an incoming arc of $x_{v_i}$ that is inside the path gadget $P_{d_i,x_{v_i}}$.
        These arcs have all been added to $S$ so the intermediate vertices of these path gadgets also form singleton strongly connected components. Note that there is no directed cycle contained within a path gadget.
        We now show that all the vertices in $G-S$ except those from $Z$ lie in the same strongly connected component.
        Consider the vertex $s$.
        As $S$ contains no edges of the imbalance gadgets we conclude that there is a strongly connected component of $G-S$ containing all the vertices in $I_{s,s_i}$ and $I_{s,d_i}$ for every $i\in [\clique]$ and $I_{s_i, x_u}$ for every $i\in [\clique]$ and every $u \in V_i$.
        Recall that the imbalance gadgets have many arcs between consecutive vertices in both directions (including the intermediate vertices).
        Now, for any $z_{uv} \notin Z$, we have either $u \notin V(K)$ or $v \notin V(K)$.
        Without loss of generality, assume that $u \notin V(K)$.
        Then, we have an arc from $x_u$ to $z_{uv}$ in $E_3 \setminus S$, so there is a path from $s$ to each $z_{uv} \notin Z$ in $G-S$.
        For each $z_{uv}$ there is a pair of arcs $z_{uv}s$ in $E_4 \setminus S$.
        Hence, for each $z_{uv} \notin Z$ there is a cycle in $G-S$ passing through $s$ and $z_{uv}$.
        Last, the intermediate vertices of path gadgets $P_{d_i,x_u}$ where $u \notin V(K)$ we see that $d_i$ and $x_u$ are in the same strongly connected component of $G-S$ due to the imbalance gadgets.
        As $P_{d_i,x_u}$ contains no edges of $S$ it follows that its intermediate vertices belong to the strongly connected component of $G-S$ that contains $d_i$ and $x_u$, which completes the proof of the claim.
    \end{claimproof}

    Building on the above claim, let us call the strongly connected component of $G-S$ containing all vertices excepting $Z$, the \emph{large} component.
    Since singleton strongly connected components are always balanced, we only need to show that the large component is Eulerian, i.e., it is balanced with respect to the active arcs.

    \begin{claim}\label{claim:large_component}
        The large component is balanced.
    \end{claim}
    \begin{claimproof}
        We consider all the vertices based on the following five cases.
        \begin{enumerate}[\bf (1)]
            \item Consider the vertex $s$, and the large component $G-Z$.
                We have that $\deg^+_G(s) = \deg^+_{G-Z}(s)$ whereas the large component contains all but $|E(H')|$ of the in-neighbors of $s$, more precisely, it is missing in-neighbors $\{z_{uw} \mid uw \in E(K)\}$.
                Recall that each vertex $z_{uw}$ has two arcs to $s$.
                Hence we have $\deg^-_{G-Z}(s) = \deg^-_G(s) - 2|E(H')|$.
                As the imbalance of the vertex $s$ in $G$ is $-2|E(H')|$ (by Claim \ref{claim:imbalance}), the vertex $s$ is balanced in the large component.

            \item The vertices in $\{s_i \mid i \in [\clique]\} \cup \{z_{uv} \mid u \notin V(K)~\text{or}~v \notin V(K)\} \cup \{ x_u \mid u \notin V(K)\}$ remain balanced as the large component contains all their in- and out-neighbors in $G$, and in $G$ these vertices were already balanced.

            \item Similarly to the previous case, we see that the vertices of all the imbalance gadgets and vertices for each $i \in [\clique]$ in path gadgets $P_{d_i,x_u}$ where $u \in V_i$ but $u \notin V(K)$ are part of the large component.
                As the intermediate vertices of these gadgets were balanced in $G$ it follows they are balanced in the large component.

            \item Now consider the vertex $d_i$ for every $i \in [\clique]$.
            From the imbalance gadgets connecting $s$ and each $d_i$, we  deduce that $d_i$ belongs to the large component $G-Z$ as well as $s$. So, $\deg^-_{G}(d_i) = \deg^-_{G-Z}(d_i)$.
                On the other hand, the large component contains all but $|N_{H'}(i)|$ many out-arcs of $d_i$ which are contained in the path gadget $P_{d_i, x_{v_i}}$.
                So we have $\deg^+_{G-Z}(d_i) = \deg^+_G(d_i) - |N_{H'}(i)|$.
                As the imbalance of the vertex $d_i$ in $G$ is $|N_{H'}(i)|$ (by Claim \ref{claim:imbalance}), the vertex $d_i$ is balanced in the large component.

            \item Finally, for every $i \in [\clique]$ consider the vertex $x_{v_i}$.
                By Claim \ref{claim:imbalance}, vertex $x_{v_i}$ is balanced in $G$.
                However, the large component does not contain all the neighbors of $x_{v_i}$.
                It excludes $|N_{H'}(i)|$ out-neighbors which are precisely $\{z_{uv_i} \mid uv_i \in E(K)\}$ and the unique in-neighbor in the path gadget $P_{d_i, x_{v_i}}$ from which there are 
                               $|N_{H'}(i)|$ arcs to $x_{v_i}$.
                Therefore, vertex $x_{v_i}$ is balanced in the large component.
        \end{enumerate}
        This completes the proof of the claim.
    \end{claimproof}

    We removed $k$ arcs and by Claim \ref{claim:connected_component} and \ref{claim:large_component} the resulting graph has balanced strongly connected components, which completes the forward direction.

    In the converse direction, assume that $(G,k)$ is a yes-instance and let $S$ be an inclusion-wise minimal solution. Let us first establish some structure of $S$, from which it will be possible to recover a solution for the original {\PSIshort} instance.

    Let $\mathcal C$ denote the strongly connected component of $G-S$ that contains $s$.
    Due to Lemma \ref{lem:imbalance_gadget}, $S$ does not contain any arcs of any of the imbalance gadgets, implying that $\mathcal C$ contains $s_i$ and $d_i$ for every $i \in [\clique]$ as well as $x_u$ for every $u \in V(G')$.
    Note that in the imbalance gadgets there are $k+1$ arc-disjoint paths from $x_u$ to $d_i$ for every $i \in [\clique]$ and $u \in V_i$, so we can use Lemma \ref{obs:path_gadget} to get that if $S$ contains arcs of a path gadget $P_{d_i,x_u}$, then they form a cut in it.
  We assume that if a cut of a path gadget $P_{d_i,x_u}$ is in $S$, then the cut consists of the incoming-arcs of $x_u$ in the gadget.
  This simplifying assumption can be made since all inclusion-wise minimal cuts of a path gadget are of the same cardinality and adding any minimal cut of a path gadget to $S$ makes all arcs of the path gadget inactive in $G-S$.

    Recall from Claim \ref{claim:imbalance} that the only imbalanced vertices in $G$ are $\{s\} \cup \{ d_i \mid i \in [\clique] \}$. Let us make some observations based on the fact that these vertices are eventually balanced in $\mathrm{strong}(G-S)$.

    For each $i \in [\clique]$ no in-arcs of $d_i$ are in $S$, because they lie in an imbalance gadget (Lemma \ref{lem:imbalance_gadget}).
    In order to make $d_i$ balanced $S$ must contain a cut of exactly one of the path gadgets starting at $d_i$, call it $P_{d_i,x_{v_i}}$.
    Recall that $x_{v_i}$ was originally balanced in $G$.
    Further, recall $x_{v_i}$ is in $\mathcal C$ along with $s_i$ and $d_i$.
    Since the imbalance gadget starting at $s_i$ and ending at $x_{v_i}$ cannot intersect $S$ and we have deleted all of the $|N_{H'}(i)|$ incoming arcs to $x_{v_i}$ from the path gadget $P_{d_i,x_{v_i}}$, the imbalance of $-|N_{H'}(i)|$ created at $x_{v_i}$ needs to be resolved by making exactly $|N_{H'}(i)|$ of its out-arcs in $E_3$ inactive in $G-S$.
    Since we have already spent the budget of $\sum_{i \in [\clique]}|N_{H'}(i)| = 2|E(H')|$ in the path gadgets, the remaining budget to be used for resolving imbalances at $\{x_{v_i} \mid i\in [\clique]\}$ is $2|E(H')|$.

    Recall that $s$ is imbalanced in $G$ and to make it balanced we need to make $2|E(H')|$ incoming arcs of $s$ (from $E_4$) inactive in $G-S$, because all outgoing arcs of $s$ lie in imbalance gadgets and cannot be in $S$.

    Finally, recall that for each $uv \in E(G')$, the vertex $z_{uv}$ is balanced in $G$ (by Claim \ref{claim:imbalance}).
    Since the strongly connected component $\mathcal C$ in $G-S$ contains the vertices $s, x_u, x_v$ (i.e., all neighbors of $z_{uv}$), for the vertex $z_{uv}$ to remain balanced in $\mathrm{strong}(G-S)$, we have the following exhaustive cases regarding the inactive arcs between vertices $s, x_u, x_v, z_{uv}$.
    \begin{enumerate}[\bf (1)]
        \item all of the four arcs incident to $z_{uv}$ are active;
        \item exactly one incoming and one outgoing arc is inactive;
        \item both incoming arcs or both outgoing arcs are inactive.
    \end{enumerate}
    As previously noted, we still need $2|E(H')|$ arcs in $E_3$ and $2|E(H')|$ arcs in $E_4$ to be inactive in $G-S$.
    The required number of inactive arcs in $E_3 \cup E_4$ is twice the remaining budget, so for every $z_{uv},x_u,x_v$, the arcs between $s, x_u,x_v,z_{uv}$ must be in Case (1) or Case (3).
    Moreover, whenever Case (3) occurs, we may assume without loss of generality that the arcs in $S$ are the two arcs $(x_u,z_{uv})$ and $(x_v,z_{uv})$. Thus, there are exactly $|E(H')|$ vertices $z_{uv}$ such that the arcs between $s, x_u,x_v,z_{uv}$ are in Case (3).

    We now extract the solution $K$ for $(G', H', \phi)$ by taking, for each $i\in [\clique]$, the vertex $v_i \in V(G')$ such that a cut of $P_{d_i,x_{v_i}}$ is contained in $S$.
    We have shown that there are exactly $|E(H')|$ vertices $z_{uv}$ such that the arcs between $s, x_u,x_v,z_{uv}$ are in Case (3) and for each $i \in [\clique]$ and the vertex $x_{v_i}$, exactly $|N_{H'}(i)|$ of its outgoing arcs are made inactive by $S$. This can only happen if for every $ij \in E(H')$, there is a vertex $z_{v_iv_j}$, implying that $v_iv_j$ is an edge in $G'$.
\end{proof}

\subsection{W[1]-hardness of ESCAD Parameterized by Vertex Cover Number}\label{sec:escad_w1_hard_vc}

In this section, we show that ESCAD is \woh when parameterized by the vertex cover number. 
Jansen, Kratsch, Marx, and Schlotter \cite{DBLP:journals/jcss/JansenKMS13} showed that \textsc{Unary Bin Packing} is \woh when parameterized by the number of bins $h$.

\defprob{\textsc{Unary Bin Packing}}{
    A set of positive integer item sizes $x_1,\dots,x_n$ encoded in unary, a pair of integers $h$ and $b$.
}{
    Is there a partition of $[n]$ into $h$ sets $J_1,\dots,J_h$ such that $\sum_{\ell \in J_j} x_\ell \leq b$ for every $j \in [h]$?
}

Jansen et al.~\cite{DBLP:journals/jcss/JansenKMS13} also showed that assuming {\ETH}  the well-known $n^{\OO(h)}$-time algorithm is asymptotically almost optimal.

\begin{proposition}[\cite{DBLP:journals/jcss/JansenKMS13}]\label{prop:UnaryBinPackingETHLowerBound}
    There is no algorithm solving the \textsc{Unary Bin Packing} problem in $f (h) \cdot n^{o(h/ \log h)}$ time for some function $f$, where $h$ is the number of bins in the input and $n$ is the input length, unless \ETH fails.
\end{proposition}

In order to carefully handle vertex balances in our reduction, it is helpful to work with a variant of the above problem, called \textsc{Exact Unary Bin Packing}, where the inequality $\sum_{\ell \in J_j} x_\ell \leq b$ is replaced with the equality $\sum_{\ell \in J_j} x_\ell = b$. That is, in this variant, all bins get filled up to their capacity.

\begin{lemma} \label{lem:exactreduction}
    There is a polynomial reduction from \textsc{Unary Bin Packing} to \textsc{Exact Unary Bin Packing} 
    that preserves the number of bins.
\end{lemma}
\begin{proof}
    Let $\mathcal I'=\big((x_1,\dots,x_n),h,b\big)$ be an instance of \textsc{Unary Bin Packing}.
 If $b \ge \sum_{i=1}^n x_i$, then $\mathcal I'$ is trivially a yes-instance and we can return a trivial  yes-instance of \textsc{Exact Unary Bin Packing} with only one bin.
    In the same way, if $b \cdot h < \sum_{i \in [n]} x_i$, then $\mathcal I'$ is trivially a no-instance and we return a trivial no-instance of \textsc{Exact Unary Bin Packing} with only one bin. Now, suppose neither of the above cases occur.

     Note that the length of the unary encoding of $b$ is upper bounded by the total length of the unary encoding of all items $x_1,\dots,x_n$.
    Similarly, if $h \ge n$ then the instance boils down to checking whether $x_i \le b$ for every $i \in [n]$ (and producing a trivial \textsc{Exact Unary Bin Packing} instance accordingly) so we can assume that $h < n$, hence, the length of the unary encoding of $h$ is upper bounded by the total 
    length of the unary encoding of all items.
    We now construct an instance $\mathcal I$ of \textsc{Exact Unary Bin Packing} from $\mathcal I'$ by adding $h \cdot b - \sum_{i \in [n]} x_i$ one-sized items (this is non-negative because of the preprocessing steps) while not changing $b$ or $h$. 
    
    If $\mathcal I'$ is a yes-instance, then one can fill-in the remaining capacity in every bin with the unit-size items, to get a solution for  $\mathcal I$.
    Conversely, if $\mathcal I$ is a yes-instance, then removing the newly added unit-size items yields a solution for $\mathcal I'$.
    Note that since summing the items in $\mathcal{I}$ gives exactly $b \cdot h$, this implies that $|\mathcal I| = \mathcal O(|\mathcal I'|^2)$ so the instance of \textsc{Exact Unary Bin Packing}  remains polynomially bounded.
\end{proof}

\ESCADWHARDVC*

\begin{proof}
We reduce from an \textsc{Exact Unary Bin Packing} instance  $\mathcal I$ to an instance $\mathcal I^*=(G,k)$ of \ESCAD in polynomial time. Let us set $k$ to be equal to $b \cdot h (h-1)$.
We now build a graph $G$ that models the $h$ bins by $h$ copies of interconnected gadgets, one for each bin. Moreover, we will ensure that the vertices in the gadgets representing the bins form a vertex cover and  each item will be modelled by a vertex of the independent set which is the complement of this vertex cover. 

In our reduction, we use the following notation. For a pair of vertices $p,q$, a $c$-arc $(p,q)$ denotes  $c$ parallel copies of the arc $(p,q)$ and a thick arc $(p,q)$ denotes a $3k$-arc $(p,q)$. The construction of $G$ is as follows.

\begin{itemize}
    \item The vertex set of $G$ is defined to be the set $\{u_j \mid j\in [h]\}\cup \{v_j \mid j\in [h]\}\cup \{w_i \mid i\in [n]\}$.

    \item For each $j\in [h]$, we add a $b$-arc $(u_j,v_j)$, a thick arc $(u_j,v_j)$ and a thick arc $(v_j,u_j)$. We call the subgraph induced by $u_j,v_j$ and these arcs, the \emph{$b$-imbalance gadget} $B_j$. 

    \item Next, we add thick arcs $(u_j,u_{j'})$ for every $j < j'$ where $j,j' \in [h]$.

    \item Finally, for each $i \in [n]$ and $j \in [h]$, we add $x_i$-arcs $(w_i,u_j)$ and $(v_j,w_i)$.
\end{itemize}

   This concludes the construction, see \Cref{fig:vc_hardness}.     Before we argue the correctness, let us make some observations.
     
     Note that the vertices participating in the imbalance gadgets form a vertex cover of the resulting graph and their number is $2h$.
     Hence, if we prove the correctness of the reduction, we have the required parameterized reduction from \textsc{Exact Unary Bin Packing} parameterized by the number of bins to ESCAD parameterized by the vertex cover number of the graph.

     We say that a  set of arcs in $G$ \emph{cuts} a $(p,q)$ arc if it contains all parallel copies of $(p,q)$. Note that no set of at most $k$ arcs cuts a thick $(p,q)$ arc. In particular, no solution to the ESCAD instance $(G,k)$ cuts any thick arc $(p,q)$ that appears in the graph. 

    \begin{figure}[ht]
        \centering
        \includegraphics[scale=1.1]{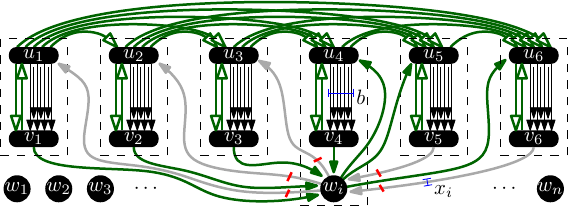}
        \caption{
            A part of the resulting \ESCAD instance after reduction from \textsc{Exact Unary Bin Packing} with six bins;
            connections between the independent vertices and imbalance gadgets are shown only for one vertex $w_i$.
            Thick arcs are shown with empty arrowhead, bold arcs incident to $w_i$ are $x_i$-arcs.
            Crossed off arcs are in a solution and dashed boxes show strongly connected components of the solution.
            This example represents $x_i \in J_4$.
        }%
        \label{fig:vc_hardness}
    \end{figure}

    \medskip
\noindent
\textbf{$\mathcal{I}$ is a yes-instance $\Rightarrow$ $\mathcal{I^*}$ is a yes-instance:}
    Assume that we have a partition $J_1,\dots,J_h$ that is a solution to $\mathcal I$. We now define a solution $S$ for $\mathcal I^*$.
    For every $x_i \in J_j$ we cut (i.e., add to $S$) all parallel copies of the arc $(w_i,u_{j'})$ for every $j' < j$ and we cut all parallel copies of the arc $(v_{j''},w_i)$ for every $j'' > j$.
    This results in cutting a total of $x_i \cdot (h-1)$ arcs incident to each $w_i$ and as $\sum_{i=1}^n x_i = b \cdot h$ we cut exactly $b \cdot h(h-1) = k$ arcs in total.

    \begin{claim}\label{claim:strong_balanced}
        $\mathrm{strong}(G-S)$ is balanced.
    \end{claim}
    \begin{claimproof}
        Due to the thick $(u_j,v_j)$ and $(v_j,u_j)$ arcs and the fact that no set of  at most $k$ arcs can cut a thick arc, we have that for every $j\in [h]$, there is a  single strongly connected component of $G-S$  containing both vertices of $B_j$.
        We next observe that in $\mathrm{strong}(G-S)$ no pair of distinct $b$-imbalance gadgets are contained in the same  strongly connected component. This is because any path in $G$ from $B_j$ to $B_i$ for $i<j$ must use arcs $(v_{j'},w_p)$ and $(w_p,u_{i'})$ for some $p\in [n]$, and $i',j'\in [h]$ such that $i'<j'$. However, one of the these two arcs is part of $S$ by definition.

        Further, notice that the strongly connected component containing $B_j$ also contains the vertex $w_i$ if $x_i \in J_j$. This is because we do not delete the arcs $(v_{j},w_i)$ and $(w_i,u_{j})$. Since we have already argued that the imbalance gadgets are all in distinct strongly connected components in $G-S$, we infer that the strongly connected component containing $B_j$ also contains the vertex $w_i$ if and only if $x_i \in J_j$. 
        Hence, we conclude that incident to $w_i$, the only active arcs are those of the form $(w_i,u_j)$ and $(v_j,w_i)$ for $j$ such that $x_i \in J_j$, making $w_i$ balanced in $\mathrm{strong}(G-S)$.
        Towards each $u_j$ and $v_j$ for $j \in [h]$, the thick $(u_j,v_j)$ and thick $(v_j,u_j)$ arcs contribute the same in-degree and out-degree.
        The only active arcs that remain and are incident to $u_j$ and $v_j$ are the $b$-arcs $(u_j,v_j)$ for each $j\in [h]$.  We argue that these are balanced by the arcs incoming from all the vertices $w_i$ to $u_j$  and the arcs outgoing from $v_j$ to $w_i$, where $x_i \in J_j$.
        Indeed, $\sum_{\ell \in J_j}{x_\ell} = b$ for all $j \in [h]$ so $u_j$ has $b$ incoming arcs and $v_j$ has $b$ outgoing arcs from the vertices $\{w_1,\dots,w_n\}$, making $u_j$ and $v_j$ balanced in $\mathrm{strong}(G-S)$ for all $j \in [h]$.
    \end{claimproof}

    \medskip
    \noindent\textbf{$\mathcal{I^*}$ is a yes-instance $\Rightarrow$ $\mathcal{I}$ is a yes-instance:}
    We aim to show that in any solution for the ESCAD instance, the arcs that are cut incident to $w_i$ for any $i \in [n]$ have the same structure as described in the other direction, i.e., for all $w_i$ there exists $j$ such that the solution cuts $(w_i,u_{j'})$ for all $j' < j$ and it cuts $(v_{j''}, w_i)$ for all $j'' > j$. This is equivalently phrased in the following claim.
    \begin{claim}\label{claim:one_cut}
        There are no two indices $a,b \in [h]$ with $a<b$ such that both $(w_i,u_a)$ and $(v_b,w_i)$ are uncut.
    \end{claim}

    \begin{claimproof}
        Towards a contradiction, consider a solution $S$ without this property. That is,  for some two indices $a,b$ where $1\leq a<b\leq h$, both $(w_i,u_a)$ and $(v_b,w_i)$ are uncut by $S$. Choose the pair $a,b$ such that $a$ is minimized.
        The graph $G$ contains thick arcs $(u_a,u_b)$ and $(u_b,v_b)$ that cannot be cut by $S$.
        Hence, there is a cycle $(w_i,u_a,u_b,v_b,w_i)$ in $\mathrm{strong}(G-S)$, implying that the vertices of two imbalance gadgets $B_a$ and $B_b$ are in the same strongly connected component of $G-S$. 
        We argue that $u_a$ cannot be balanced in $\mathrm{strong}(G-S)$.
        We first ignore the two thick arcs $(u_a,v_a)$ and $(v_a,u_a)$ as they balance each other.
        We picked $a$ to be minimum, so $u_a$ has no active incoming arcs that belong to a thick $(u_{a'},u_{a})$ arc for some $a'<a$ since otherwise, $u_a,v_a,u_{a'},v_{a'}$ would be in the same strongly connected component of $G-S$. 
        Hence, the only remaining active incoming arcs on $u_a$ are the  incoming arcs from $\{w_1,\dots,w_n\}$, of which there are exactly $\sum_{i \in [n]} x_i$ arcs. Recall that we have $\sum_{i \in [n]} x_i=b\cdot h$  and by the definition of $k$, we have $h\cdot b<k$.
        This implies that in $G-S$, $u_a$ has at least $2k$ active outgoing arcs (at most $k$ out of the $3k$ arcs contained in the thick $(u_a,u_b)$ arc can be in $S$) and at most $k$ incoming active arcs, a contradiction to $S$ being a solution. Hence, we conclude that for all $w_i$ there exists $j$ such that the solution cuts $(w_i,u_{j'})$ for all $j' < j$ and it cuts $(v_{j''}, w_i)$ for all $j'' > j$.
    \end{claimproof}

    We next argue that if $S$ is a solution, then for all $w_i$, there exists $j$ such that the solution is disjoint from any $(w_i,u_j)$ arc and any $(v_j,w_i)$ arc. 
    Since the budget is $k = b \cdot h (h-1)$ we have that: If we cut more than $x_i (h-1)$ arcs incident to $w_i$ for some $i \in [n]$, then there exists $i' \in [n] \setminus \{i\}$ such that we cut fewer than $x_{i'} (h-1)$ arcs incident to $w_{i'}$.
    But this would violate Claim \ref{claim:one_cut}.
    Hence, for any solution $S$, we can retrieve the assignment of items to bins in the \textsc{Exact Unary Bin Packing} instance $\cal I$, by identifying for every $i \in [n]$, the unique value of $j \in [h]$ such that $S$ is disjoint from any $(w_i,u_j)$ arc and any $(v_j,w_i)$ arc and then assigning item $x_i$ to bin $J_j$.
    This completes the proof that \ESCAD is \woh parameterized by the vertex cover number. 

    Since our reduction from \textsc{Exact Unary Bin Packing} to \ESCAD transforms the parameter linearly and the instance size polynomially, by invoking Lemma~\ref{lem:exactreduction} and Proposition~\ref{prop:UnaryBinPackingETHLowerBound}, we obtain the claimed \ETH-based lower bound parameterized by the vertex cover number for \ESCAD.
\end{proof}

\subsection{NP-hardness of ESCAD on Graphs of Constant Maximum Degree}\label{sec:escad_np_hard_maxdegree}

We show that \ESCAD is para-\nph~when parameterized by the maximum degree.

\ESCADNPHDEGREE*

\begin{proof}
    We give a polynomial-time reduction from {\sc Vertex Cover} on cubic (3-regular) graphs, which is known to be {\nph} \cite{DBLP:books/daglib/p/Karp10}, to \ESCAD. 
    This reduction is a modification of the proof in \cite{DBLP:books/daglib/p/Karp10} which shows that {\sc Directed Feedback Arc Set} is \nph. The input to \textsc{Vertex Cover} consists of a graph $G$ and an integer $k$; the task is to decide whether $G$ has a vertex cover of size at most $k$.
    Let $(G,k)$ be an instance of  \textsc{Vertex Cover} with $n$ vertices where $G$ is a cubic graph.  
    We construct an \ESCAD instance $\mathcal I'=(G',k)$ in the following way. The vertex set $V(G')= V(G)\times \{0,1\}$ and the arc set $E(G')$ is defined by the union of the sets 
    $\{((u,0), (u,1)) \mid u \in V(G)\}$ and $\{((u,1), (v,0))^2 \mid uv \in E(G)\}$. We  call the arcs of the form $((u,0), (u,1))$  \emph{internal arcs} and arcs  of the form $((u,1), (v,0))$  \emph{cross arcs}. Note that the notation $(x,y)^2$ in the above definition refers to two parallel arcs from $x$ to $y$.
    
     Towards the correctness of the reduction, we prove the following claim.
    
    \begin{claim}\label{claim:vc_reduction}
        $(G,k)$ is a yes-instance of {\sc Vertex Cover} if and only if $(G',k)$ is a yes-instance of \ESCAD.
    \end{claim}

    \begin{claimproof}
        In the forward direction, let $(G,k)$ be a yes-instance and  let $X$ be a solution. Consider the arc set  $F= \{((u,0),(u,1)) \mid u\in X\} \subseteq E(G')$.  We  show that $F$ is a feedback arc set of $G'$. Consider any cycle in $G'$. Due to our construction, the cycle must contain two internal arcs $((u,0), (u,1))$ and $((v,0), (v,1))$ where $uv \in E(G)$. Now either $u \in X$ or $v \in X$. That implies that either $((u,0), (u,1))$ or $((v,0), (v,1))$ belongs to $F$. Hence $G' -F$ has no cycles. As $G'-F$ is acyclic, we have that $F$ is a solution to the \ESCAD instance $(G',k)$.

        In the converse direction, let $(G',k)$ be a yes-instance, let $F$ be a solution for this instance  with minimum number of cross arcs. 
        We first argue that $G'-F$ is acyclic. Suppose not.  
        Because of the structure of the constructed digraph, every cycle alternates between internal and cross arcs. So, every strongly connected component of $G'-F$ that is not a singleton must contain an internal arc, and as it must also be Eulerian, the strongly connected component must be a simple cycle $C$ (as for each $u \in V(G)$, the out-degree of the vertex $(u,0)$ in $G' -F$  is at most one).
        Each arc of $C$ is present only once -- to achieve that, the solution $F$ must contain at least one copy of each of the cross arcs that are in $C$. Now, we can remove all the copies of cross arcs in $C$ from the solution and instead, add all internal arcs of $C$ to the solution. This gives us a new solution with fewer cross arcs, a contradiction to our choice of $F$. Hence, we may assume that $G'-F$ is acyclic.   We now argue that $X= \{u \mid ((u,0), (u,1)) \in F\} \cup \{u \mid ((u,1), (v,0)) \in F\}$ is a vertex cover of $G$ of size at most $k$. Clearly $|X| \leq k$.  Consider an arbitrary edge $uv \in E(G)$. Corresponding to the edge $uv$ there is a  4-cycle $((u,0), (u,1))$, $((u,1), (v,0)),$ $ ((v,0), (v,1))$, $((v,1), (u,0))$ in $G'$, and so,  $F$ must contain one of these four arcs. 
        Now, by our definition of $X$, $X \cap \{u,v\} \neq \emptyset$, hence $X$ is a solution for the {\sc Vertex Cover} instance $(G,k)$.
    \end{claimproof}

    This shows that \ESCAD is \nph. Moreover, Since $G$ is  a cubic graph, every vertex in $D'$ has (in, out) degree equal to $(1,6)$ or $(6,1)$.
    This completes the  proof of \Cref{theo:nph}.
\end{proof}
\subsection{Algorithms for ESCAD on Graphs of Bounded Treewidth}\label{sec:escad_tw_tractability}
Due to \Cref{thm:escad_whard_vc}, the existence of an {\fpt} algorithm for \ESCAD parameterized by standard width measures such as treewidth is unlikely. This raises the following questions: (i) Could we get {\fpt} algorithms if we parameterize by treewidth {\em and} some additional parameters of the input? (ii) Could we get an XP algorithm parameterized by treewidth? (iii) Could one obtain  algorithms whose running times match our {\ETH}-based lower bounds? 

In this section, we give  an algorithm for {\ESCAD} that is simultaneously, an
{\sf XP} algorithm parameterized by treewidth, an
{\fpt} algorithm parameterized by the treewidth and solution size, and also an {\fpt} algorithm parameterized by the treewidth and maximum degree of the input digraph. 
Let us note that in the specific case of parameterizing by treewidth and maximum degree, if all we wanted was an {\fpt} algorithm, then we could use Courcelle's theorem at the cost of a worse running time. However, our  algorithm in one shot gives us the three results mentioned above and in simple digraphs, achieves nearly optimal asymptotic dependence on the treewidth assuming {\ETH} (see Theorem \ref{thm:combinedSESCAD} for the lower bound statement).

\ESCADTWXP*

Since the maximum degree is upper bounded by the instance length as already discussed in \Cref{sec:prelims}, this gives an \XP algorithm parameterized by treewidth alone. However, when in addition to treewidth we parameterize either by the size of the solution or by the maximum degree this gives an \fpt algorithm.

\begin{corollary}
    \ESCAD is \fpt parameterized by $\tw + k$, \fpt parameterized by $\tw + \Delta$, and \XP parameterized by $\tw$ alone.
\end{corollary}

Recall that in digraphs, multiarcs are permitted.
So, we use a variant of the nice tree decomposition notion.
This is defined for a digraph $G$ by taking a nice tree decomposition with introduce edge nodes (see \Cref{sec:prelims}) of the simple undirected graph underlying $G$ then expanding each introduce edge node to introduce all parallel copies of arcs one by one. Note that although the new introduce arc nodes introduce \emph{arcs}, the orientation does not affect the decomposition. Let us call such a tree decomposition of $G$, a {\em nice tree decomposition with introduce arc nodes}. 

Korhonen and Lokshtanov~\cite{DBLP:conf/stoc/KorhonenL23} gave a   $2^{tw^2} \cdot n^{\OO(1)}$-time algorithm that computes an optimal tree decomposition, which can be used in our algorithm.  We could also use any of the constant-factor approximation algorithms in \cite{DBLP:journals/jct/RobertsonS95b,DBLP:journals/siamcomp/BodlaenderDDFLP16,DBLP:conf/focs/Korhonen21} instead. We note that this would come at the cost of blowing up the exponents of $\alpha$ in  Theorem \ref{thm:escad_tw_xp} by a constant multiplicative factor. 
However, in terms of asymptotics, this would not make a difference to our running time. 
Moreover, we use the fact that any tree decomposition can be converted to a nice tree decomposition of the same width with introduce edge nodes in polynomial time \cite{DBLP:books/sp/CyganFKLMPPS15}, and the introduce edge nodes can clearly be expanded to introduce arc nodes in polynomial time.
Since the running time of our algorithm dominates the time taken for this step, we may assume that we are given such a tree decomposition. In addition, we insist that once an arc $(u,v)$ is introduced, all parallel copies of this arc and parallel copies of the arc $(v,u)$ (if they exist) are consecutively introduced. Let $G_t$ be the subgraph of the input graph that contains the vertices and arcs introduced in the subtree rooted at $t$.
We refer to $G_t$ as the {\em past} and to all arcs and vertices not present in $G_t$ as the {\em future}.

We express the reachability of the graph that lies outside (both past and future) of the ``current'' bag during our dynamic program using the following notation. 

\begin{definition}
   {\em  For a set $X$, the tuple $(R,\ell)$ is a \emph{reachability arrangement on $X$} if $R$ is a simple digraph with $V(R) = X$, and $\ell$ is a labeling where $\ell\colon E(R) \to \{\mathrm{direct}, \mathrm{past}, \mathrm{future}\}$. We use $\ell(u,v)$ to denote $\ell((u,v))$.
}
\end{definition}

\begin{definition}
	{\em 
	For a digraph $G$ and a vertex $v\in V(G)$, we define the {\em active out-arcs of $v$ in $G$} as the subset of the set $\{(v,u)\mid u\in V(G)\}$ comprising those arcs that lie in the same strongly connected component of $G$ as $v$.  
	The {\em active out-degree of $v$ in a subgraph $H$ of $G$} is the number of active out-arcs of $v$ in $G$ that are contained in $H$. 
	The {\em active in-arcs of $v$ in $G$} and {\em active in-degree of $v$ in a subgraph $H$ of $G$} are defined symmetrically.   The {\em active imbalance of $v$ in a subgraph $H$ of $G$} is obtained by subtracting the active in-degree of $v$ in the subgraph $H$ of $G$ from the active out-degree of $v$ in the subgraph $H$ of $G$.  We drop the explicit reference to $G$ when it is clear from the context.
	}
\end{definition}

The aforementioned notion of active imbalance of $v$ in a subgraph $H$ of $G$ will later enable us to make meaningful statements regarding the ``contribution'' of certain subsets of arcs towards the ultimate  imbalance of a vertex in its strongly connected component when considering the original graph minus a hypothetical solution.

\begin{remark}[Overview of our algorithm]
{\em 
We present a dynamic programming algorithm over tree decompositions.
When one attempts to take the standard approach, the main challenge that arises is that through disconnecting strongly connected components, removing an arc can affect vertices ``far away'' from each other and hence possibly impact vertices that have already been forgotten at the current stage of the algorithm.
Our workaround for this obstacle is to guess the partition of each bag into strongly connected components in the final solution and then keep track of the imbalances of the vertices of the bag under this assumed partition. This allows us to safely forget a vertex as long as its ``active'' imbalance is zero.
Another difficulty lies in keeping track of how these assumed connections interact with the bag: whether they use vertices already forgotten or those yet to be introduced.

For our dynamic programming algorithm(s) over the tree decomposition $(\mathcal{T}, \{X_t\}_{t \in V(\mathcal{T})})$ to work, we also need to ensure that for every bag $X_t$ and every vertex $v\in X_t$, we are able to produce a bounded (in the parameters under consideration) set of possible values of the active imbalance of $v$ in the subgraph $G_t-S$ of $G-S$, where $S$ is a solution. When the maximum degree of $G$, denoted $\Delta$,  is a parameter, this set of values is trivially obtained since the imbalance of any vertex in any subgraph of $G$ is never more than twice the maximum degree of $G$.  For the case where our parameters are $k$ and treewidth, this is more involved as the {\em values} of the possible active imbalances of a vertex  could be arbitrarily large compared to the parameters. However, we prove a lemma (see Lemma \ref{lem:smallRange} later) that shows that we only need to consider boundedly many possible values for the active imbalance of any vertex $v$ in the subgraph $G_t-S$ of $G-S$.
 Moreover, we can efficiently compute this set of possible values for the active imbalance of $v$ in $G_t-S$.
}
\end{remark}

Before we describe our algorithm for {\ESCAD}, we first  prove the lemma bounding the number of possible values for the active imbalance of $v$ in $G_t-S$ as mentioned above. Towards this, we prepare as follows.

\begin{definition}
	{\em 
	For a digraph $G$ and disjoint vertex sets $S$ and $T$, an $S$-$T$ path in $G$ is a path that starts at a vertex in $S$, ends at a vertex in  $T$ and whose internal vertices are disjoint from $S\cup T$. When $S$ or $T$ is a singleton, we simply write it as a vertex instead of  as a  singleton set (e.g., we write  $s$-$t$ path instead of $\{s\}$-$\{t\}$ path). 
	}
\end{definition}
\begin{definition}\label{def:flows}
	{\em  For a digraph $G$ and disjoint vertex sets $S$ and $T$, an $S$-$T$ flow is a set $\PP$ of pairwise arc-disjoint $S$-$T$ paths in $G$. The {\em value} of the flow $\PP$ is the size of the set $\PP$ and the paths in $\PP$ are simply called $S$-$T${\em flow paths} when $\PP$ is clear from the context and simply {\em flow paths} when $S$ and $T$ are also implied.	
	}
\end{definition}

The classic Ford-Fulkerson algorithm \cite{fordFulkerson65} enables one to compute a maximum $S$-$T$ flow in polynomial-time.   

\begin{proposition}\label{prop:fordFulkerson}
	Given a digraph $G$, disjoint vertex sets $S$ and $T$, an $S$-$T$ flow of maximum value can be computed in polynomial time.
\end{proposition}

Let us now define the notion of {\em circulations} using which will be able to express useful properties of  Eulerian digraphs.

\begin{definition}\label{def:circulation}
{\em	For a digraph $G$ and a vertex $v\in V(G)$, we define a {\em $v$-circulation} in $G$ to be a set $\CC$ of pairwise arc-disjoint cycles in $G$ that each contain $v$. The {\em value} of the circulation is the size of the set $\CC$.
We say that a $v$-circulation $\CC$ {\em contains} an arc $a$ if $a$ appears in some element of $\CC$.
 }
\end{definition}

\begin{observation}\label{obs:activeArcsInCirculations}
	For a digraph $G$ and arc $a=(u,v)\in E(G)$, if both endpoints of $a$ lie in the same strongly connected component of $G$, then there is a $v$-circulation that contains $a$. 
\end{observation}

Let us make another simple observation giving an upper bound on the value of a $v$-circulation.

\begin{observation}\label{obs:trivial}
	For a digraph $G$ and a vertex $v\in V(G)$, the maximum value of a $v$-circulation in $G$ is at most the minimum of the in-degree and out-degree of $v$.
\end{observation}

The following lemma says that on Eulerian digraphs, the trivial upper bound stated above is achieved. 

\begin{lemma}\label{lemma:eulerianCirculation}
	If $G$ is an Eulerian digraph, then for every $v\in V(G)$, there is a $v$-circulation in $G$ whose value is equal to the out-degree (equivalently, the in-degree) of $v$. 
\end{lemma}

\begin{proof}
Let $\gamma$ denote the value of the out-degree (and in-degree) of $v$. 
Let $W=u_1,\dots,u_\ell$ be an Euler tour in $G$.
Precisely, (i) $W$ is a closed walk, that is, each $u_i\in V(G)$, $u_1=u_\ell$ and for every $i\in [\ell-1]$, $(u_i,u_{i+1})$ is an arc in $G$, and (ii) every arc in $G$ appears exactly once in $W$. 
We assume without loss of generality that $u_1=v$. 
Let $1	=i_1<\dots<i_r=\ell$ denote exactly those indices such that $u_{i_j}=v$ for each $j\in [r]$.  
That is, these are the indices where the walk revisits $v$. 
Since we started with an Euler tour, we have that $\gamma=r-1$.
Consider the set of walks $\mathcal{W}=\{W_{i_j}\mid j\in [r-1]\}$ where $W_{i_j}$ is defined as the subwalk of $W$ starting at $u_{i_j}$ and ending at $u_{i_{j+1}}$. By construction, each walk in $\mathcal{W}$ is a closed walk containing $v$ and so, each walk $W_{i_j}$ contains a cycle $C_{i_j}$ passing through $v$. Moreover, since $W$ is an Euler tour, the walks in $\mathcal{W}$ are pairwise arc-disjoint, implying that the cycles in the set $\{C_{i_j}\mid j\in [r-1]\}$ are also pairwise arc-disjoint. As already established, we have that $\gamma=r-1$, and thus, we have obtained the required $v$-circulation.
\end{proof}

We next define a notion of ``partial'' circulations that will be used in our algorithm.

\begin{definition}\label{def:generalCirculation}
{\em For a digraph $G$, $v\in V(G)$, $Z\subseteq V(G)\setminus \{v\}$, we define a {\em $(v,Z)$-circulation in} $G$ as a set $\CC=\{C_1,\dots,C_r\}$, where:
\begin{enumerate}
	\item Each $C_i$ is either a cycle passing through $v$ and disjoint from $Z$ or is a path from $v$ to some $z\in Z$ and whose  internal vertices are disjoint from $Z$; and  
	\item for every $i\neq j\in [r]$, $C_i$ and $C_j$ are arc-disjoint. 
\end{enumerate}	
Symmetrically, we define a {\em $(Z,v)$-circulation in $G$} as a set $\CC=\{C_1,\dots,C_r\}$, where:
\begin{enumerate}
	\item Each $C_i$ is either a cycle passing through $v$ and disjoint from $Z$ or is a path to $v$ from some $z\in Z$ and whose  internal vertices are disjoint from $Z$; and  
	\item for every $i\neq j\in [r]$, $C_i$ and $C_j$ are arc-disjoint. 
\end{enumerate}	
The {\em value} of a $(v,Z)$-circulation or a $(Z,v)$-circulation is defined as the size of the respective set.  We say that a $(v,Z)$-circulation $\CC$ {\em contains} an arc-set $S$ if for every arc $a\in S$,  there is an element of $\CC$ containing $a$.
}
\end{definition}

\begin{observation}
	A $(v,\emptyset)$-circulation or a $(\emptyset,v)$-circulation in $G$ is also a $v$-circulation in $G$.
\end{observation}

We next introduce a definition that will ultimately enable us to relate (partial) circulations in $G$ to flows in an auxiliary digraph. Recall that the digraphs we work with in this paper do not have self-loops.

\begin{definition}\label{def:augmentedDigraph}
{\em 	For a digraph $G$ and $v\in V(G)$, we define the graph $G(v)$ as the digraph obtained from $G$ by performing the following operations. 
	\begin{enumerate}
		\item Add vertices $v^{+}$, $v^{-}$. 
		\item For every arc $a=(v,u)$ in $G$ where $u\in V(G)$, we add the arc $a^{+}=(v^{+},u)$. We say that $a^{+}$ is the image of $a$ in $G(v)$ and $a$ is the pre-image of $a^{+}$ in $G$. 
		\item  For every arc $a=(u,v)$ in $G$ where $u\in V(G)$, we add the arc $a^{-}=(u,v^{-})$. 
		We say that $a^{-}$ is the image of $a$ in $G(v)$ and $a$ is the pre-image of $a^{-}$ in $G$. 
		\item Finally, delete the vertex $v$. 
	\end{enumerate}
	Effectively, we ``split'' the vertex $v$ into two copies, one capturing all out-going arcs from $v$ and the other capturing all in-coming arcs to $v$.
	}
\end{definition}

We have the following observation summarizing a key property of the above definition.

\begin{observation}\label{obs:propertyOfAugmentedGraph}
Consider a digraph $G$, $v\in V(G)$ and $Z\subseteq V(G)\setminus \{v\}$.
\begin{enumerate}\item For every cycle $C$ in $G$ containing $v$, replacing the two arcs of $C$ incident to $v$ with their images in $G(v)$ results in a $v^{+}$-$v^{-}$ path in $G(v)$. Conversely, for every $v^{+}$-$v^{-}$ path in $G(v)$, replacing the arcs incident to  $v^{+}$ and $v^{-}$ with their pre-images in $G$ results in a cycle  containing $v$ in $G$. 
\item 	 For every $v$-$Z$ path in $G$, replacing the arc incident to $v$ with its image in $G(v)$ results in a $v^{+}$-$Z$ path in $G(v)$ whose internal vertices are disjoint from $Z$. Conversely, for every   $v^{+}$-$Z$ path in $G(v)$ whose internal vertices are disjoint from $Z$, replacing the arc incident to $v^{+}$ with its pre-image in $G$ results in a $v$-$Z$ path in $G$.  
\item For every $Z$-$v$ path in $G$, 
replacing the arc incident to $v$ with its image in $G(v)$ results in a $Z$-$v^{-}$ path in $G(v)$. 
Conversely, for every   $Z$-$v^{-}$ path in $G(v)$, 
replacing the arc incident to $v^{-}$ with its pre-image in $G$ results in a $Z$-$v$ path in $G$.

\end{enumerate}
\end{observation}

We are now ready to formally relate $(Z,v)$-circulations in $G$ to flows in $G(v)$.

\begin{lemma}\label{lem:correspondenceToFlows}
Consider a digraph $G$, $v\in V(G)$ and $Z\subseteq V(G)\setminus \{v\}$. Let $\gamma\in {\mathbb N}$. 
\begin{enumerate}\item  There is a $(v,Z)$-circulation of value $\gamma$ in $G$ if and only if there is a $v^{+}$-$(\{v^{-}\}\cup Z)$ flow of value $\gamma$ in $G(v)$. \item There is a $(Z,v)$-circulation of value $\gamma$ in $G$ if and only if there is a $(\{v^{+}\}\cup Z)$-$v^{-}$ flow of value $\gamma$ in $G(v)$.
\end{enumerate}

\end{lemma}

\begin{proof}
We prove the first statement of the lemma. The second statement is proved symmetrically. 

Let $\CC$ be a $(v,Z)$-circulation of value $\gamma$ in $G$. By Observation \ref{obs:propertyOfAugmentedGraph}~(1), for 
each cycle in $\CC$, replacing the arcs incident to $v$ with their images in $G(v)$ leads to a $v^{+}$-$v^{-}$ path in $G(v)$. Moreover, since such cycles are disjoint from $Z$ by the definition of $(v,Z)$-circulations, the resulting $v^{+}$-$v^{-}$ paths are also disjoint from $Z$.  By Observation \ref{obs:propertyOfAugmentedGraph}~(2),  for each $v$-$Z$ path in $\CC$, replacing the arc incident to  $v$ with its image in $G(v)$ leads to a $v^{+}$-$Z$ path in $G(v)$. Moreover, since each $v$-$Z$ path in $\CC$ has its internal vertices disjoint from $Z$ by the definition of $(v,Z)$-circulations, it follows that the resulting $v^{+}$-$Z$ path also has its internal vertices disjoint from $Z$. Finally, since the elements of $\CC$ are originally pairwise arc-disjoint and the images in $G(v)$ of distinct arcs in $G$ are distinct by construction, we have thus obtained a $v^{+}$-$(\{v^{-}\}\cup Z)$ flow of value $\gamma$ in $G(v)$.

Conversely, consider a $v^{+}$-$(\{v^{-}\}\cup Z)$ flow $\PP$ of value $\gamma$ in $G(v)$. For each flow path $P\in \PP$ from $v^{+}$ to $v^{-}$, we replace the arcs incident to $v^{+}$ and $v^{-}$ with their pre-images in $G$ to obtain a cycle in $G$ that contains $v$, by Observation \ref{obs:propertyOfAugmentedGraph}~(1). Moreover the cycle resulting in this way is disjoint from $Z$ since the flow path has its internal vertices disjoint from $Z$ by the definition of flows (Definition \ref{def:flows}).   For each flow path $P\in \PP$ from $v^{+}$ to $Z$, replacing the arc incident to $v^{+}$ with its pre-image in $G$  leads to a $v$-$Z$ path $P'$ in $G$ by Observation \ref{obs:propertyOfAugmentedGraph}~(2). Moreover, by definition, $P$ has its internal vertices disjoint from $Z$ and so, $P'$ also has its internal vertices disjoint from $Z$. Since the flow paths in $\cal P$ are arc-disjoint and the pre-images in $G$ of distinct arcs in $E(G(v))\setminus E(G)$ are distinct by construction, we have obtained the required $(v,Z)$-circulation in $G$. 
\end{proof}

Let us next demonstrate how the values of partial circulations help us estimate the active out- and in-degree of a vertex during our dynamic program. We will use the term {\em deletion set for a digraph $G$} to refer to any set $S$ of arcs of $G$ such that every strongly connected component of $G-S$ is Eulerian.

\begin{lemma}\label{lem:partialActiveDegrees}
Consider a  tree decomposition $(\mathcal{T}, \{X_t\}_{t \in V(\mathcal{T})})$ of a digraph $G$ and let $t\in V(\TT)$. Let $S$ be a deletion set for $G$.  
	 Let $v\in X_t$ and $Z\subseteq X_t\setminus\{v\}$ such that $\{v\}\cup Z$ is the intersection of $X_t$ with the vertex set of some strongly connected component of $G-S$. 

\begin{enumerate}\item 	   
	
\begin{enumerate}\item 
Let $\Gamma^{+}$ denote the set of active out-arcs  of $v$ in $G-S$ that are contained in $G_t$ and let $\gamma^{+}$ denote the size of this set. Then, there is a $(v,Z)$-circulation of value $\gamma^{+}$ in $G_t-S$, containing the arcs in $\Gamma^{+}$.
\item Let $\Gamma^{-}$ denote the set of active in-arcs of $v$ in $G-S$ that are contained in $G_t$ and let $\gamma^{-}$ denote the size of this set. Then, there is a $(Z,v)$-circulation of value $\gamma^{-}$ in $G_t-S$, containing the arcs in $\Gamma^{-}$. 
\end{enumerate}
\item 	
\begin{enumerate}\item Consider a $(v,Z)$-circulation $\CC=\{C_1,\dots,C_{\gamma^{+}}\}$ in $G_t$ and let $a_1,\dots,a_{\gamma^{+}}$ be such that $a_i$ is the unique out-arc leaving $v$ in $C_i$, for each $i\in [\gamma^{+}]$. 
At least $\gamma^{+}-|S|$ of the arcs in $\{a_i\mid i\in [\gamma^{+}]\}$ are active out-arcs of $v$ in $G-S$.
\item Consider a $(Z,v)$-circulation $\CC=\{C_1,\dots,C_{\gamma^{-}}\}$ in $G_t$ and let $a_1,\dots,a_{\gamma^{-}}$ be such that $a_i$ is the unique in-arc entering $v$ in $C_i$, for each $i\in [\gamma^{-}]$. 
At least $\gamma^{-}-|S|$ of the arcs in $\{a_i\mid i\in [\gamma^{-}]\}$ are active in-arcs of $v$ in $G-S$.
\end{enumerate}
\end{enumerate}

\end{lemma}

\begin{proof}
We only prove statements 1~(a) and 2~(a). The proofs for 1~(b) and 2~(b) are symmetric. 

Consider 1~(a). 	By Lemma \ref{lemma:eulerianCirculation}, there is a $v$-circulation $\CC$ in $G-S$ whose value is exactly the active out-degree of $v$ in $G-S$. By definition of $v$-circulations, this implies that there exist $\gamma^{+}$ elements of $\CC$ in which the unique out-arc leaving $v$ is an arc in $E(G_t)$. Let $C_1,\dots,C_{\gamma^{+}}$ denote these elements of $\CC$. We now define a set $\CC'=\{C'_1,\dots, C'_{\gamma^{+}}\}$ as follows. For each $C_i$, if it is a cycle in $G_t$ disjoint from $Z$, then set $C'_i:=C_i$. Next, consider any $C_i$  that is not a cycle in $G_t$ and traverse this cycle starting from $v$ and following the unique out-arc in $C_i$ at each subsequent vertex. Since $C_i$ is not contained in $G_t$ and $X_t$ is a bag in the tree decomposition, it must be the case that there is a vertex $z\in X_t$ that is contained in $C_i$ such that the subpath $P_i$ of $C_i$ from $v$ to $z$ is contained in $G_t$ and whose internal vertices are disjoint from $X_t$. Since every vertex contained in an element of the $v$-circulation $\CC$ is in the same strongly connected component of $G-S$ as $v$, the premise of the lemma defining $Z$ guarantees that in fact, $z\in Z$ and 
the internal vertices of $P_i$ are disjoint from $Z$.
We set $C'_i:=P_i$. Finally, consider any $C_i$ that is a cycle in $G_t$ intersecting $Z$. By repeating the same argument as above (i.e., traversing the cycle out of $v$), we can obtain a subpath $P_i$ of $C_i$ from $v$ to some $z\in Z$ such that $P_i$ is contained in $G_t$ and its internal vertices are disjoint from $Z$. We set $C'_i:=P_i$. 
Now, it is straightforward to check that $\CC'$ satisfies the requirements of a $(v,Z)$-circulation in $G_t$ containing the arcs in $\Gamma^{+}$.

	Consider 2~(a).	 First note that for any $i\in [\gamma^{+}]$, if $C_i$ is disjoint from $S$ then $a_i$ either lies in a cycle containing $v$ in $G-S$ (if $C_i$ is a cycle in $G_t$) or on a $v$-$Z$ path in $G-S$ (if $C_i$ is a path in $G_t$). By assumption, $\{v\}\cup Z$ lies in a strongly connected component of $G -S$, implying that if $C_i$ is disjoint from $S$, then $a_i$ is an active out-arc of $v$ in $G-S$. 
The statement then follows from the fact that the elements of $\CC$ are pairwise arc-disjoint and so, at least $\gamma^{+}-|S|$ of these elements are disjoint from $S$. 
\end{proof}

We are now ready to combine our observations thus far to produce a small set of values capturing the imbalance of each vertex in a bag ``imposed'' by the graph below this bag. Let us note that for a subgraph $H$ of $G$ and arc set $S\subseteq E(G)$, the notation $H-S$ refers to the subgraph $H-(S\cap E(H))$. 

\begin{lemma}\label{lem:smallRange}
For a given tree decomposition $(\mathcal{T}, \{X_t\}_{t \in V(\mathcal{T})})$ of a digraph $G$, $t\in V(\TT)$, a reachability arrangement $(R,\ell)$ on $X_t$, a vertex $v\in X_t$ and $k\in {\mathbb N}$,  there is a polynomial-time computable set $\lambda(t,R,v,k)\subseteq {\mathbb Z}$ of size at most $2k+1$ with the property that if $(G,k)$ is a yes-instance of {\ESCAD} with a solution $S$ such that the partition of $X_t$ among the strongly connected components of $G-S$ is the same as the partition of $X_t$ among the strongly connected components of $R$, then the active imbalance of $v$ in the subgraph $G_t-S$ of $G-S$ is contained in $\lambda(t,R,v,k)$. 
\end{lemma}

\begin{proof}
Let $R(v)$ denote the subset of vertices of $R$ other than $v$ that lie in the same strongly connected component of $R$ as $v$. 
Compute a maximum $v^{+}$-$(\{v^{-}\}\cup R(v))$ flow in  the graph $G_t(v)$ using Proposition \ref{prop:fordFulkerson} and let $\gamma^{+}$ be its value. Compute a maximum $(\{v^{+}\}\cup R(v))$-$v^{-}$ flow in  the graph $G_t(v)$ using Proposition \ref{prop:fordFulkerson} and let $\gamma^{-}$ be its value. Denote $\gamma^{+}-\gamma^{-}$ by $\gamma^{*}$ and define $\lambda(t,R,v,k):=\{\gamma^{*}-k,\dots,\gamma^{*}+k\}$.  

As the size of the produced set clearly satisfies the stated bound, it remains to argue the correctness. Let $S$ be a solution for $(G,k)$ satisfying the conditions in the premise of the lemma. By  Lemma \ref{lem:partialActiveDegrees}~(1)~(a),  the number of active out-arcs of $v$ in $G-S$ that are contained in $G_t$ is at most the value of the maximum $(v,R(v))$-circulation in $G_t$ and by Lemma \ref{lem:partialActiveDegrees}~(2)~(a),  at least  the value of the maximum $(v,R(v))$-circulation in $G_t$ minus $k$. By Lemma \ref{lem:correspondenceToFlows} (1), we have that the maximum value of a  $(v,R(v))$-circulation in $G_t$ is the same as the maximum value of a $v^{+}$-$(\{v^{-}\}\cup R(v))$ flow in $G_t(v)$, i.e., $\gamma^{+}$.
Similarly, by Lemma \ref{lem:partialActiveDegrees}~(1)~(b),  the number of active in-arcs of $v$ in $G-S$ that are contained in $G_t$ is at most the value of the maximum $(R(v),v)$-circulation in $G_t$ and by Lemma \ref{lem:partialActiveDegrees}~(2)~(b) at least the value of the  maximum $(R(v),v)$-circulation in $G_t$ minus $k$. By Lemma \ref{lem:correspondenceToFlows} (2), we have that the maximum value of a $(R(v),v)$-circulation in $G_t$ is the same as the maximum value of a $(\{v^{-}\}\cup R(v))$-$v^{+}$ flow in $G_t(v)$, i.e., i.e., $\gamma^{-}$.  So, the active out-degree of $v$ in the subgraph $G_t-S$ of $G-S$ lies between $\gamma^{+}-k$ and $\gamma^{+}$ while the active in-degree of $v$ in the subgraph $G_t-S$ of $G-S$ lies between $\gamma^{-}-k$ and $\gamma^{-}$. 
So, the correctness follows. 
\end{proof}

Let us now generalize the above lemma by also including the maximum degree of the digraph as a parameter. 

\begin{lemma}\label{lem:smallRangeFinal}
For a given tree decomposition $(\mathcal{T}, \{X_t\}_{t \in V(\mathcal{T})})$ of a digraph $G$ with maximum degree $\Delta$,  $t\in V(\TT)$, a reachability arrangement $(R,\ell)$ on $X_t$, a vertex $v\in X_t$ and $k\in {\mathbb N}$ there is a polynomial-time computable set $\Lambda(t,R,v,k)\subseteq {\mathbb Z}$ of size at most ${\rm  min}\{2k+1,2\Delta+1\}$ with the property that if $(G,k)$ is a yes-instance of {\ESCAD} with a solution $S$ such that the partition of $X_t$ among the strongly connected components of $G-S$ is the same as the partition of $X_t$ among the strongly connected components of $R$, then the active imbalance of $v$ in the subgraph $G_t-S$ of $G-S$ is contained in $\Lambda(t,R,v,k)$. 
\end{lemma}

\begin{proof}
If $\Delta\leq k$, then we set $\Lambda(t,R,v,k)=\{-\Delta,\dots,\Delta\}$. Otherwise, we set $\Lambda(t,R,v,k)$ to be the set $\lambda(t,R,v,k)$ obtained by invoking Lemma \ref{lem:smallRangeFinal}. 
The correctness is straightforwad since the out-degree and in-degree of the vertex $v$ even in the original graph $G$ are both bounded by $\Delta$. 
\end{proof}

Equipped with the above structural result, we are now ready to start describing the necessary components of our dynamic program.

\begin{definition}\label{def:consistency}
{\em     For a given tree decomposition $(\mathcal{T}, \{X_t\}_{t \in V(\mathcal{T})})$ of a digraph $G$, 
a node $t$ of the tree decomposition and a reachability arrangement $(R, \ell)$ on $X_t$, we call a set of arcs $S \subseteq E(G_t)$ \emph{consistent} with $(R, \ell)$ if the following parts hold.

    \begin{enumerate}
        \item\label{part:direct} An arc $(u,w) \in \ell^{-1}(\mathrm{direct})$ if and only if $(u,w)$ is an arc in $G_t[X_t] - S$.
        \item\label{part:past_realized} An arc $(u,w) \in \ell^{-1}(\mathrm{past})$ if and only if 
        $(u,w) \notin E(G_t - S)$
        and there is a path from $u$ to $w$ in $G_t - S$ that contains no vertices from $X_t \setminus \{u,w\}$ (also called {\em path through the past}).
        \item\label{part:future_realized} An arc $(u,w) \in \ell^{-1}(\mathrm{future})$ if and only if
            $(u,w) \notin E(G_t - S)$,
            there is no path through the past from $u$ to $w$,
            and there is a path from $u$ to $w$ in $G - S$ that contains no vertices from $X_t \setminus \{u,w\}$ (also called {\em path through the future}).
            \item For every pair $u,w\in X_t$ such that $(u,w)$ is not an arc in $R$, there is no $u$-$w$ path in $G_t-S$ whose internal vertices are disjoint from $X_t$.
    \end{enumerate}
    }
\end{definition}

\begin{definition}
{\em 	 For a given tree decomposition $(\mathcal{T}, \{X_t\}_{t \in V(\mathcal{T})})$ of a digraph $G$,  node $t$ of the tree decomposition, a reachability arrangement $(R, \ell)$ on $X_t$,  
 a subgraph $G'$ of $G_t$, and a vertex $v\in V(G')$,  the {\em $R$-augmented active imbalance of $v$ in the graph $G'$} is the active imbalance of $v$ in the subgraph $G'$ of $G'+E(R)$.}
\end{definition}

In the above definition, there is a slight abuse of notation since $V(G')$ may not include some vertices in $X_t$. In these cases, $G'+E(R)$ refers to the graph obtained from $G'$ by adding only those arcs of $R$ that have both endpoints in $G'$. However, we will typically use this notation when $G'$ is a spanning subgraph of $G_t$ and this issue will not arise in such cases.

\begin{remark}{\em The idea behind the above definition is the following. 
The arcs in $R$ are essentially ``fake arcs'', introduced specifically to represent reachability relations between vertices in $X_t$ within the graph $G$ minus a hypothetical solution. Their purpose is only to determine whether other arcs in $G_t$ contribute to the eventual active imbalance of a vertex and they themselves do not factor into the computation of active imbalance. So for this we introduce the notion of $R$-augmented active imbalance where the arcs in $R$ determine the strongly connected components, but do not necessarily contribute to the active imbalance of vertices. 
}
\end{remark}

\begin{lemma}\label{lem:inducedActiveImbalance}
 For a given tree decomposition $(\mathcal{T}, \{X_t\}_{t \in V(\mathcal{T})})$ of a digraph $G$, node $t\in V(\TT)$, a reachability arrangement $(R,\ell)$ on $X_t$ and a vertex $v\in X_t$, suppose that $(G,k)$ is a yes-instance of {\ESCAD} with a solution $S$ such that the partition of $X_t$ among the strongly connected components of $G-S$ is the same as the partition of $X_t$ among the strongly connected components of $R$. 
 Then, the active imbalance of $v$ in the subgraph $G_t-S$ of $G-S$ is equal to the $R$-augmented active imbalance of $v$ in the graph $G_t-S$.
\end{lemma}

\begin{proof}
Let $R(v)$ denote the subset of vertices of $R$ other than $v$ that lie in the same strongly connected component of $R$ as $v$. By assumption, $\{v\}\cup R(v)$ is the intersection of $X_t$ with the vertex set of some strongly connected component of $G-S$. 

Towards the proof of the lemma, we prove the following claim.
\begin{claim}
	For every arc $a\in E(G_t-S)$, $a$ is an active out-arc of $v$ in $G-S$ if and only if $a$ is an active out-arc of $v$ in the graph $G_t-S+R$.
\end{claim}

\begin{claimproof}
Let $\Gamma^{+}$ denote the set of active out-arcs of $v$ in $G-S$ that are contained in $G_t$ and let $\gamma^{+}$ be its size. 
By Lemma \ref{lem:partialActiveDegrees}~(1)~(a), there is a $(v,R(v))$-circulation $\CC$ of size $\gamma^{+}$  in $G_t-S$, that contains $\Gamma^{+}$. Since the vertices of $\{v\}\cup R(v)$ are in the same strongly connected component of $R$ by definition, 
we have that the arcs in $\Gamma^{+}$ are active in $G_t-S+R$. 

Conversely, let $a\in E(G_t-S)$ be an active out-arc of $v$ in the graph $G_t-S+R$. Then, it is 
contained in some $v$-circulation in $G_t-S+R$ (Observation \ref{obs:activeArcsInCirculations}). This implies that there is a cycle $C_a$ in $G_t-S+R$ containing $a$. If $C_a$ is contained in $G-S$, then we are done since we have shown that $a$ is active in $G-S$. So, suppose not. In particular, $C_a$ is not contained in $G_t-S$. However, since it is present in $G_t-S+R$, 
there must be a path in $G_t-S$ (and hence, also in $G-S$) starting in $v$, containing $a$ and ending at some vertex $r\in R(v)$. By the definition of $R(v)$, $v$ and $r$ are in the same strongly connected component of $G-S$, implying that $a$ is an active out-arc of $v$ in $G-S$.  This completes the proof of the claim. 
\end{claimproof}

A symmetric argument to the above gives us the following claim. 
\begin{claim}For every arc $a\in E(G_t-S)$, $a$ is an active in-arc of $v$ in $G-S$ if and only if $a$ is an active in-arc of $v$ in the graph $G_t-S+R$. 
	\end{claim}

From the above two claims, we have that the number of active out-arcs of $v$ in $G_t-S+R$ minus the number of active in-arcs of $v$ in $G_t-S+R$ (i.e., the $R$-augmented active imbalance of $v$ in $G_t-S$) is the same as the number of active out-arcs of $v$ in $G-S$ that are contained in $G_t$ minus the number of active in-arcs of $v$ in $G-S$ that are contained in $G_t$ (i.e., the active imbalance of $v$ in the subgraph $G_t-S$ of $G-S$). This completes the proof of the lemma. 
\end{proof}

We next formally describe when a partial solution is eligible to be stored at a specific index of our dynamic programming table. This is done over the following two definitions where we fix  a tree decomposition $(\mathcal{T}, \{X_t\}_{t \in V(\mathcal{T})})$ of a digraph $G$. 

\begin{definition}\label{def:wellFormedTuple}
{\em     Given a node $t$ of the tree decomposition, 
we say that  $(t,R, \ell, b, W,k)$ is a {\em well-formed} tuple if the following hold:
\begin{enumerate}
\item $(R,\ell)$ is a reachability arrangement on $X_t$, 
\item $b\colon X_t \to {\mathbb N}$ where $b(v)\in \Lambda(t,R,v,k)$ for every $v\in X_t$,
\item  $W \subseteq E(G_t[X_t])$, 
\item if $(u,v)\in W$ and $\ell(u,v)={\rm direct}$, then 
there should be at least two parallel arcs $(u,v)$ in  $E(G_t[X_t])$ and at least one of these arcs is not contained in $W$, 
\item $|W|\leq k\in {\mathbb N}$. 
\end{enumerate}
}
\end{definition}

\begin{definition}\label{def:reachability_arrangement}
{\em     Given a node $t$ of the tree decomposition, we call a set of arcs $S \subseteq E(G_t)$ \emph{compatible} with a well-formed tuple $(t,R, \ell, b, W,k)$ if the following hold:
    \begin{enumerate}
        \item\label{part:bag_solution} $S$ agrees with $W$ on $G_t[X_t]$, that is $S \cap E(G_t[X_t]) = W$, 
        \item $S$ is consistent with $(R,\ell)$, 
        \item\label{part:bag_offsets} For each vertex $u \in X_t$, the the $R$-augmented active imbalance of $u$ in $G_t - S$ is $b(u)$,  
        \item\label{part:past_balanced} For each vertex $u \in V(G_t) \setminus X_t$, the $R$-augmented active imbalance of $u$ in $G_t - S$ is zero,  
        \item $|S|\leq k$. 
    \end{enumerate}
    }
\end{definition}

We are now ready to complete the proof of \Cref{thm:escad_tw_xp}.

\begin{proof}[Proof of \Cref{thm:escad_tw_xp}]
Let $(G,k)$ be the given instance. If $\alpha\leq 1$, then we are done by brute-force computation (when $k\leq 1$) or $G$ is already acyclic (when $\Delta\leq 1	$), so we may assume that $\alpha\geq 2$. 
Let $(\mathcal{T}, \{X_t\}_{t \in V(\mathcal{T})})$ be the computed nice tree decomposition of $G$ with introduce arc nodes. For every well-formed tuple $\tau=(t,R,\ell,b,W,k)$, we denote by $A[t,R,\ell,b,W,k]$ the minimum size of an arc subset of $G_t$ that is compatible with $\tau$; if there is no arc subset of $G_t$ that is compatible with $\tau$, then $A[t,R,\ell,b,W,k]=\infty$, where $\infty$ denotes a prohibitively high value for our solution, say, the number of arcs in $G$ plus one. To keep the notation simple, we will use $A[\tau]$ to denote $A[t,R,\ell,b,W,k]$.

In our decomposition $(\mathcal{T}, \{X_t\}_{t \in V(T)})$ the root node $r$ has $X_r = \emptyset$ and $G_r = G$ so $A[r,\emptyset,\emptyset,\emptyset,\emptyset,k]$ is the minimum size of an arc subset $S$ of $G$ that is compatible with $(r,\emptyset,\emptyset,\emptyset,\emptyset,k)$, which boils down to every vertex of $G$ having active imbalance of zero in $G-S$ and $|S|\leq k$. Hence, 
$A[r,\emptyset,\emptyset,\emptyset,\emptyset,k]$ is equal to the minimum size of a solution for our instance $(G,k)$ of {\ESCAD} (if one exists), or $\infty$ if no solution exists.
In order to compute $A[r,\emptyset,\emptyset,\emptyset,\emptyset,k]$ we employ the standard approach of bottom-up dynamic programming over our tree  decomposition.

For leaf nodes $X_t = \emptyset$, hence, the graphs and labelings are also empty and the empty arc set is vacuously compatible with them $A[t,\emptyset,\emptyset,\emptyset,\emptyset,k] = 0$.

For every non-leaf node $t$ and graph $R$ on $X_t$ we first calculate the strongly connected components of $R$ and also the set $\Lambda(t,R,v,k)$ for every $v\in X_t$ using Lemma \ref{lem:smallRangeFinal}. 
Then for each $\ell$, $b$, and $W$ such that $\tau=(t,R,\ell,b,W,k)$ is a well-formed tuple, we calculate $A[\tau]$ based on the type of the node $t$. When doing so, we may assume that $A[\tau']$ has already been computed for every well-formed tuple $\tau'=(t',R',\ell',b',W',k)$ where $t'$ is a descendant of $t$ in $\mathcal{T}$. 
Note that we always restrict our attention to well-formed tuples as only such tuples index our table. Moreover, by Definition \ref{def:wellFormedTuple}  and Lemma \ref{lem:smallRangeFinal}, it is straightforward to generate all well-formed tuples. The time required to do this is dominated by the time taken to fill the table, which is analyzed towards the end of this proof.  

\begin{description}
    \item[Introduce vertex node:] 
    Consider a well-formed tuple $\tau=(t,R,\ell,b,W,k)$. 
    When $t$ is an introduce vertex node and its child is $t'$ with $X_t = X_{t'} \cup \{v\}$ we know that $v$ will be isolated in $G_t$ (i.e., has degree zero) so we can effectively ignore (i.e., set the value to be $\infty$ for) entries $A[t,R,\ell,b,W,k]$ where  the reachability arrangement $(R,\ell)$ includes direct or past arcs incident to $v$. So, we may assume that no direct or past arcs incident to $v$ exist in $(R,\ell)$. For the same reason ($v$ is isolated in $G_t$), the demanded $R$-augmented active imbalance of $v$ in $G_t$ minus any set compatible with the tuple $\tau$ must be zero, i.e., $b(v)$ must be zero. If $\tau$ does not satisfy this property, then we set $A[\tau]=\infty$. 
    
    Let us next identify those entries $A[t',R',\ell',b',W',k]$ from which we will derive the value of $A[\tau]$. 
    Any future arcs incident to $v$ in $(R,\ell)$ and contributing to ``future paths'' between vertices of $X_{t'}$ should be reflected in the relevant reachability arrangement on $X_{t'}$, that is, if $(R,\ell)$ contains a future arc from $u$ to $v$ and one from $v$ to $w$ there should be a future arc from $u$ to $w$ in $(R',\ell')$, unless there is already an arc from $u$ to $w$ in $R$. 
    No arcs were introduced or forgotten at node $t$, so the set $W$ remains the same, that is, $W'=W$. 
    The formal description of the recursive formula follows.

Given a reachability arrangement $(R, \ell)$ on $X_t$, let $(R', \ell')$ be the reachability arrangement induced by $(R,\ell)$ on $X_{t'}$ except that for each pair of vertices $u, w \in X_{t'}$ where $(u, w) \notin E(R)$, $(u,v) \in E(R)$, and $(v,w) \in E(R)$ such that $\ell(u,v)=\ell(v,w)=\mathrm{future}$ we have $(u, w) \in E(R')$ and $\ell'(u,w) = \mathrm{future}$.
\[
A[t,R,\ell,b,W,k] =
    \begin{cases}
        \infty &\text{if there exists } u \in X_{t'} \text{ such that} \\
                & \qquad (u,v) \in E(R) \text{ and } \ell(u,v) \neq \textrm{future}, \\
        \infty &\text{if there exists } u \in X_{t'} \text{ such that} \\
                & \qquad (v,u) \in E(R) \text{ and } \ell(v,u) \neq \textrm{future}, \\
        \infty &\text{if } b(v) \neq 0,\\
        A[t',R',\ell',b|_{X_{t'}},W,k] &\text{otherwise}.
    \end{cases}
\]

Clearly this entry can be calculated in polynomial time given the previous table entries.

    \smallskip
    
    \item[Introduce arc node:]
     Consider a well-formed tuple $\tau=(t,R,\ell,b,W,k)$.
    Assume $t$ introduces arc $a=(u,v)$ and its child node is $t'$. Consider $A[t,R,\ell,b,W,k]$. Let us next identify the entries $A[t',R',\ell',b',W',k]$ from which we will derive the value of $A[\tau]$. We distinguish between two cases based on whether or not the new arc $a$ belongs to $W$. 
    \begin{itemize}
    \item[Case 1:] The case where $a\in W$ is straightforward. We account for $a$ being required to be contained in the compatible set sought for, by adding 1 to the entry in  $A[t',R,\ell,b,W\setminus \{a\},k]$ after making sure the size constraint of $k$ is satisfied. 
\item[Case 2:]     In the case where $a\notin W$, we ensure that $\ell(u,v)$ is direct (if not, assign value $\infty$) and then distinguish subcases based on the label of $(u,v)$ in the reachability arrangement in the indices we look up. Note that since $a\notin W$ and we aim to find the size of a smallest set compatible with $\tau$, we may assume that $(u,v)$ is an arc in $(R',\ell')$, so it is only the label that could be different.

Let us now formalize this. For each $\iota \in \{\mathrm{past}, \mathrm{direct}, \mathrm{future}\}$, let $\ell_\iota$ be the function such that $\ell_\iota(a) = \iota$ and $\ell_\iota(e) = \ell(e)$ for all other $e \in E(R)$. If the strongly connected components of $R$ containing $u$ and $v$ are the same, then define $b'(u) = b(u) - 1$, $b'(v) = b(v) + 1$, and $b'(w) = b(w)$ for all $w \in X_{t'} \setminus \{u,v\}$; otherwise, define $b'=b$.
\end{itemize}

We are now ready to give the recursive definition of $A[t,R,\ell,b,W,k]$ for the introduce arc node $t$. 

\[
    A[t,R,\ell,b,W,k] = \begin{cases}
        A[t',R,\ell,b,W \setminus \{a\},k] + 1 &\text{if } a \in W \text{ and }A[t',R,\ell,b,W \setminus \{a\},k]\leq k-1\\ 
        \min_\iota A[t',R,\ell_{\iota},b',W,k] & \text{over all } \iota \in \{\mathrm{past}, \mathrm{direct}, \mathrm{future}\} \\
                                    & \text{if } (u,v) \in E(R) \text{ and } \ell(u,v) = \mathrm{direct} \text{ and }a\notin W,\\
        \infty &\text{otherwise}. \\ 
    \end{cases}
\]

This entry for the node $t$ can be calculated in polynomial time given all the table entries for the child node $t'$.

 \item[Forget node:]
  Consider a well-formed tuple $\tau=(t,R,\ell,b,W,k)$, where $t$ is a forget node with child $t'$ such that $X_t = X_{t'} \setminus \{v\}$.
    
    Let us next identify the entries $A[t',R',\ell',b',W',k]$ from which we will derive the value of $A[\tau]$. 
%
Essentially, we require that the tuples indexing these entries  are comprised of ``legal'' extensions of $(R,\ell)$, $b$ and $W$ and all these extensions involve $v$. 
This is done by ensuring the following hold. 
    
    \begin{itemize}
    	\item Since $v$ is forgotten at node $t$, the required $R'$-augmented active imbalance of $v$ in $G_{t'}$ minus a hypothetical compatible set, i.e., $b'(v)$, should be zero (see Definition \ref{def:reachability_arrangement} (4)).
    	\item  There are no future arcs in $(R',\ell')$ incident to $v$ (otherwise $v$ could not have been forgotten at this point). 
    	\item $W'$ contains the same arcs of $G_t[X_t]$ as $W$.  
    	\item Finally, $(R',\ell')$ should reflect the possibility that past arcs in $(R,\ell)$ appear due to paths in  $G_t$ that pass through $v$ and have no internal vertex from $X_t$. That is, for every  $(u,v),(v,w) \in E(R')$ such that either $(u,w) \notin E(R')$ or $\ell'(u,w) \neq \mathrm{direct}$, we must have that $(u,w) \in E(R)$ and $\ell(u,w) = \mathrm{past}$. 

    \end{itemize} 
    
Let us now formalize the above.     Let $\mathcal{R}$ be the set of reachability arrangements $(R', \ell')$ on $X_{t'}$ such that
\begin{enumerate}
\item if $(u,v)$ is not an arc in $R$ for $u,v\in X_t$, then $(u,v)$ is not an arc in $R'$, 
    \item for every arc $e \in E(R')$ incident to $v$ we have $\ell'(e) \neq \mathrm{future}$,
\item for every arc $e\in E(R)$ such that $\ell(e)$ is either direct or future, 
$e$ is an arc in $R'$ and $\ell'(e)=\ell(e)$,  
\item for every arc $e\in E(R)$ such that $\ell(e)$ is past, either $e\in E(R')$ and  $\ell'(e)$ is past or $e\notin E(R')$ and $(u,v),(v,w) \in E(R')$. 
\end{enumerate}

Let $\mathcal Q$ be the set of quadruplets $(R',\ell',b',W')$ such that $(R',\ell') \in \mathcal{R}$, 
$b'|_{X_t} = b$ and $b'(v) = 0$,  
$W'\cap E(G_t[X_t]) = W$ and $|W'|\leq k$. 
If $\mathcal Q=\emptyset$, then we set $A[t,R,\ell,b,W,k]=\infty$ and otherwise, 
\[
    A[t,R,\ell,b,W,k] = \min{\{A[t',R',\ell',b',W',k] \mid (R',\ell',b',W') \in \mathcal Q \}}.
\]

We have no more than $4^{(\tw+1)^2}$ reachability arrangements in $\mathcal{R}$ which can be easily iterated through by brute-force,
$b'$ is uniquely determined by $b$ and there are no more than minimum of $\{(\Delta+1)^{2(\tw+1)},(k+1)^{2(\tw+1})\}$ possibilities for $W'$ since it can contain at most minimum of $\{\Delta,k\}$ arcs between $v$ and each vertex of $X_t$ and there are two possible orientations for each such arc, besides which it is identical to $W$.
Hence, we can compute $A[t,R,\ell,b,W,k]$ in time $2^{\bigoh(\tw^2)} \cdot \alpha^{\bigoh(\tw)} \cdot  n^{\OO(1)}$.

    \item[Join node:]
    Consider a well-formed tuple $\tau=(t,R,\ell,b,W,k)$ where $t$ is a join node with children $t_1$ and $t_2$. We will go through various pairs of well-formed tuples $\tau_1$ and $\tau_2$ where $\tau_i=(t_i,R_i,\ell_i,b_i,W_i,k)$, look up  $A[\tau_1]$ and $A[\tau_2]$ and derive $A[\tau]$ from their values.
     
    We first observe that the reachability arrangements $(R_i,\ell_i)$ should be chosen such that $R=R_1=R_2$. Moreover, past arcs in  $(R,\ell)$ can appear either as past arcs in both of $\{(R_i,\ell_i)\mid i\in [2]\}$ or we can have a past arc in one arrangement while there is a future arc in the other arrangement.
    In a similar way, we need to consider for each $u \in X_t$ how the $R$-augmented active imbalance $b(u)$ in $G_t-S$ (where $S$ is a hypothetical minimum size set compatible with $\tau$) is made up of parts in $G_{t_1}-S$ and $G_{t_2}-S$ while making sure not to double count the contributions of arcs in $G_t[X_t]$.
    Finally, we require $W=W_1=W_2$ and subject to the above, $A[\tau]$ would simply be the sum of $A[\tau_1]$ and $A[\tau_2]$  minus $|W|$ to avoid double counting the arcs in $W$. Let us now formalize this.

Let $L_{R, \ell}$ be the set of pairs of functions $\ell_1, \ell_2$ such that for all $e \in E(R)$
\begin{enumerate}
    \item if $\ell(e) = \mathrm{past}$, then $(\ell_1(e),\ell_2(e)) \in \{(\mathrm{past},\mathrm{past}),(\mathrm{future},\mathrm{past}),(\mathrm{past},\mathrm{future})\}$,
    \item otherwise $\ell_1(e) = \ell_2(e) = \ell(e)$.
\end{enumerate}
 
For each $u\in X_t$, let us use $\rho(u)$ to denote the $R$-augmented active imbalance of $u$ in the graph $G_t[X_t]-W$. Note that $\rho(u)$ isolates the contribution of the direct arcs (and their parallel copies) towards the $R$-augmented active imbalance of $u$ in $G_{t_i}-S$ for each $i\in [2]$, and we must take care not to double count it when calculating active imbalances. Given $\tau$, it is straightforward to compute the function $\rho:X_t\to {\mathbb Z}$. 

Let $b_1$ and $b_2$ be functions such that for each $u \in X_t$ we have
\[
    b(u) = b_1(u) + b_2(u) - \rho(u).
\]
Let $\mathcal B$ be the set of $(b_1,b_2)$ pairs that conform to the above equality and where, in addition, $b_i(u)\in \Lambda(t_i,R,u,k)$ for each $u\in X_t$. If ${\mathcal B}=\emptyset$ or $L_{R,\ell}=\emptyset$, then we set $A[t,R,\ell,b,W,k]=\infty$. Otherwise, define

\begin{align*}
   \eta
 = \min \{ &A[t_1,R,\ell_1,b_1,W,k] + A[t_2,R,\ell_2,b_2,W,k] - |W| \\
                            &:(\ell_1,\ell_2) \in L_{R,\ell}, (b_1,b_2) \in \mathcal B \}.
\end{align*}

If $\eta\leq k$, then set $A[t,R,\ell,b,W,k]=\eta$, otherwise set $A[t,R,\ell,b,W,k]=\infty$. 

Let us address the computation time for this join node. 
$L_{R, \ell}$ contains at most $3^{(\tw+1)^2}$ pairs of functions.
The pairs in $\mathcal B$ comprise functions that are defined over a range of size at most $2\alpha+1$ and for a fixed $b$ there are at most $2\alpha+1$ ways to choose $b_1(u)$, each of which fixes $b_2(u)$. 
As we iterate over the values of these functions applied to each $u \in X_t$ independently, there are at most $(2\alpha+1)^{\tw+1}$ ways to choose a suitable pair of functions $b_1$ and $b_2$.
The minimum is taken over $L_{R, \ell}$ and $\mathcal B$ so this entry can be calculated in $2^{\bigoh(\tw^2)}\cdot  \alpha^{\bigoh(\tw)} \cdot n^{\OO(1)}$ time.

\end{description}

\paragraph{Overall running time analysis.}
We know that the total number of nodes in the nice tree decomposition with introduce arc nodes is $n^{\OO(1)}$ and it can be observed that this still holds for the extension to multiarcs.
For a fixed node $t$ and well-formed tuple where the first coordinate is $t$, there are $4^{(\tw+1)^2}$ possible reachability arrangements on $X_t$, $(2\alpha+1)^{\tw+1}$ choices of $b$, and at most $2^{(\tw+1)^2}$ choices of $W$ on simple digraphs and at most $\alpha^{(\tw+1)^{2}}$ choices of $W$ in general. 

\begin{itemize}\item So, we have $4^{(\tw+1)^2}\cdot (2\alpha+1)^{\tw+1}\cdot \alpha^{(\tw+1)^{2}}\cdot n^{\bigoh(1)}$ table entries to fill in general, which is bounded by  $\chi_1(\alpha,\tw,n)=\alpha^{\bigoh(\tw^{2})}\cdot n^{\bigoh(1)}$. 
	
\item For simple digraphs, we have $4^{(\tw+1)^2}\cdot (2\alpha+1)^{\tw+1}\cdot 2^{(\tw+1)^{2}}\cdot n^{\bigoh(1)}$ table entries to fill, which is bounded by  $\chi_2(\alpha,\tw,n)=2^{\bigoh(\tw^{2})}\cdot \alpha^{\bigoh(\tw)}\cdot n^{\bigoh(1)}$.
 \end{itemize}

This is the only place where our analysis distinguishes between simple digraphs and general digraphs. 

Both introduce vertex nodes and introduce arc nodes compute each entry
 in $n^{\OO(1)}$ time.
Each $A[\tau]$ where the first coordinate of $\tau$ is a forget node is computed in time 
$2^{\bigoh(\tw^2)}\cdot  \alpha^{\bigoh(\tw)} \cdot n^{\OO(1)}$,  while each $A[\tau]$ where the first coordinate of $\tau$ is a join node is computed in time $\chi_3(\alpha,\tw,n)=2^{\bigoh(\tw^2)}\cdot  \alpha^{\bigoh(\tw)} \cdot n^{\OO(1)}$. As the time taken for join nodes dominates the time required for the other nodes, we can bound the overall running time of our algorithm by $\chi_1(\alpha,\tw,n)\cdot \chi_3(\alpha,\tw,n)$ in general and by  $\chi_2(\alpha,\tw,n)\cdot \chi_3(\alpha,\tw,n)$ on simple digraphs, giving the bounds stated in the theorem. 
\end{proof}

\section{Our Results for \ESCAD on Simple Digraphs}
In this section, we study \ESCAD on simple digraphs, which we formally define as follows.

\defprob{\textsc{Simple Eulerian Strong Component Arc Deletion (\SESCAD)}}{
    A simple digraph $G$, an integer $k$
}{
    Is there a subset $R \subseteq E(G)$ of size $|R| \le k$ such that in $G - R$ each strongly connected component is Eulerian?
}

Let us begin by stating a simple observation that enables us to make various inferences regarding the complexity of \SESCAD based on the results we have proved for \ESCAD.

\begin{observation}\label{obs:equivalence}
    Consider an \ESCAD instance $\mathcal I=(G,k)$. 
   If we subdivide every arc $(u,v)$ into $(u,w),(w,v)$ (using a new vertex $w$) then we get an equivalent \SESCAD instance $\mathcal I'=(G',k)$ with $|V(G')|=|V(G)|+|E(G)|$ and $|E(G')|=2|E(G)|$.
\end{observation}

\subsection{Hardness Results for \SESCAD}

We first discuss the implications of \Cref{theo:escadsolution,thm:escad_whard_vc,theo:nph} for \SESCAD along with Observation \ref{obs:equivalence}.

\begin{corollary}\label{theo:sescadsolution}
    \SESCAD is \woh when parameterized by the solution size.
\end{corollary}

\begin{proof}
    Follows from \Cref{theo:escadsolution} and Observation \ref{obs:equivalence}. 
\end{proof}

\begin{observation}\label{obs:vcToModtostars}
    If we subdivide all arcs in a digraph $G$ that has a vertex cover $X$, we get a simple digraph $G'$ such that $G'-X$ is the disjoint union of directed stars.
\end{observation}

\begin{corollary}\label{cor:vcToModtostars}
    \SESCAD is \woh parameterized by minimum modulator size to disjoint union of directed stars.
\end{corollary}

Under \ETH, we have the following result.

\begin{restatable}{theorem}{SESCAD_ETH_LB}
There is no algorithm solving \SESCAD in $f(k)\cdot n^{o(k/ \log k)}$ time for some function $f$, where $k$ is the size of the smallest vertex set that must be deleted from the input graph to obtain a disjoint union of directed stars and $n$ is the input length, unless the Exponential Time Hypothesis fails.
\end{restatable}

\begin{proof}
The reduction in the proof of \Cref{thm:escad_whard_vc} along with Proposition  \ref{prop:UnaryBinPackingETHLowerBound} and Observation \ref{obs:equivalence} and Observation \ref{obs:vcToModtostars} implies the statement.
\end{proof}

Note that the above result rules out an {\fpt} algorithm for \SESCAD parameterized by various width measures such as treewidth and even treedepth.

\begin{restatable}{theorem}{NPHSIMPLE}\label{theo:nphsimple}
    \SESCAD is \nph~in  simple digraphs where each vertex has $(\mathrm{in}, \mathrm{out})$ degrees from the set $\{(1,1), (1,6), (6,1)\}$.
\end{restatable}

\begin{proof}
    Follows from Theorem \ref{theo:nph} and Observation \ref{obs:equivalence}.
\end{proof}

\subsection{{\fpt} Algorithms for \SESCAD}

Firstly, the {\fpt} algorithms discussed in the previous section naturally extend to \SESCAD.
However, for \SESCAD, the lower bound parameterized by modulator to a disjoint union of  directed stars leaves open the question of parameterizing by larger parameters. For instance, the vertex cover number.

To address this gap, we provide an \fpt algorithm for \SESCAD parameterized by {\em vertex integrity}, a~parameter introduced by Barefoot et al.~\cite{Barefoot1987}.

\begin{definition}[Vertex Integrity]\label{def:vertex_integrity}
{\em     An undirected graph $G = (V,E)$ has \emph{vertex integrity} $k$ if there exists a set of vertices $M \subseteq V$
of size at most $k$ such that when removed, each connected component of $G-M$ has size at most $k-|M|$. We call $M$ a {\em certificate of vertex integrity} $k$ for $G$. 
    We define the vertex integrity of a directed graph $G'$ to be the vertex integrity of the undirected graph underlying it. 
 Similarly, given a digraph $G'$, we say that a set $M\subseteq V(G')$  is a certificate of vertex integrity $k$ for $G'$ if it is a certificate of vertex integrity $k$ for the undirected graph underlying $G'$.} 
\end{definition}

{\fpt} algorithms parameterized by vertex integrity have gained popularity in recent years due to the fact that several problems known to be {\woh} parameterized even by treedepth can be shown to be {\fpt} when parameterized by the vertex integrity \cite{DBLP:journals/tcs/GimaHKKO22}. Since Corollary \ref{cor:vcToModtostars}  rules out {\fpt} algorithms for {\SESCAD} parameterized by treedepth, it is natural to explore {\SESCAD} parameterized by vertex integrity and our positive result thus adds \SESCAD to the extensive list of problems displaying this behavior.

Moreover, this {\fpt} algorithm parameterized by vertex integrity implies that \SESCAD is also \fpt  when parameterized by the vertex cover number and shows that our reduction for \ESCAD parameterized by the vertex cover number requires multiarcs for fundamental reasons and cannot be just adapted to simple digraphs with more work.

\newcommand{\ilpfeas}{{\sc ILP-Feasibility}\xspace}
We will use as a subroutine the well-known \fpt  algorithm for \ilpfeas.
The \ilpfeas problem is defined as follows.
The input is a  matrix $A\in {\mathbb Z}^{m\times p}$ and a vector $b\in {\mathbb Z}^{m \times 1}$ and the objective is to find a vector $\bar x\in {\mathbb Z}^{p \times 1}$ satisfying the $m$ inequalities given by $A$, that is, $A\cdot \bar x\leq b$, or decide that such a vector does not exist.

\begin{proposition}[\cite{DBLP:journals/mor/Lenstra83,DBLP:journals/mor/Kannan87,DBLP:journals/combinatorica/FrankT87}]\label{prop:lenstra}
    \ilpfeas can be solved
    using $\OO(\eta^{2.5\eta+o(\eta)}\cdot L)$ arithmetic
    operations and space polynomial in $L$,
    where $L$ is the number of bits in the input and $\eta$ is the number of variables.
\end{proposition}

\begin{definition}
Consider a directed graph $G$ on $n$ vertices. 
A function $\pi:V(G)\to [n]$ is called an {\em ordering} of the vertices of  $G$. For a set $X\subseteq V(G)$, the ordering on $X$ {\em induced by} $\pi$ is the function $\pi_X:X\to [|X|]$, such that for every $u,v\in X$, $\pi_X(u)<\pi_X(v)$ if $\pi(u)<\pi(v)$. 
The {\em lexicographic ordering on $V(G)\times V(G)$ induced by} $\pi$ is a function $\chi:V(G)\times V(G)\to [n^{2}]$ such that for every two distinct, ordered pairs $e_1=(u_1,v_1)$ and $e_2=(u_2,v_2)$, $\chi(e_1)<\chi(e_2)$ if: 
\begin{itemize}\item $\pi(u_1)<\pi(u_2)$; or \item $\pi(u_1)=\pi(u_2)$ and $\pi(v_1)<\pi(v_2)$. \end{itemize}

\end{definition}

For a tuple $\alpha$, we refer to the $i^{\rm th}$ element of $\alpha$ by $\alpha[i]$. 

\begin{definition}
	Consider a directed graph $G$ and a function $\pi:V(G)\to [n]$. Let $\chi$ be the lexicographic ordering on $V(G)\times V(G)$ induced by $\pi$. 
	The {\em edge-characteristic vector of $G$ induced by} $\pi$ is a tuple $\vec{a}\in \{0,1\}^{n^{2}}$ such that for every pair $e=(u,v)$:
	\begin{itemize}
		\item $\vec{a}[\chi(e)] =0$ if $(u,v) \notin E(G)$, 
		\item $\vec{a}[\chi(e)] =1$ otherwise.
	\end{itemize}
\end{definition}

Note that the edge-characteristic vector of $G$ induced by $\pi$ is, in essence, the same as the adjacency matrix of $G$ where the rows and columns are ordered according to $\pi$. However, to simplify the notation in the upcoming proof, we work with vectors.

For a digraph $G$ and sets $X\subseteq V(G)$ and $Y\subseteq E(G[X])$, we denote by $G[X,Y]$, the subgraph of $G$ with vertex set $X$ and arc set $Y$. 

\sescadVIFPT*

\begin{proof}
    Consider an instance $(G,p)$ of {\SESCAD}, where $G$ has vertex integrity at most $k$. Suppose that this is a yes-instance with a solution $S$. 
    That is, $S$ is an arc-set of size at most $p$ and in $G-S$, each strongly connected component is Eulerian.
     Let $M$ be a certificate of vertex integrity $k$ for $G$. We do not require $M$ to be given in the input since there is an {\fpt} algorithm parameterized by $k$ to compute it \cite{DBLP:journals/algorithmica/DrangeDH16}.
           Without loss of generality, assume that $V(G)=[n]$ and $M=[|M|]$. This fixes an ordering $\pi$ over $V(G)$, in which the vertices of $M$ appear first.  In the rest of the proof, for every set $X\subseteq V(G)$, we use $\pi_X$ to refer to $\pi$ induced on $X$. 

    We begin by guessing those arcs of $S$ that have both endpoints in $M$, remove them and reduce the budget $p$ accordingly. The number of possible guesses is $2^{\OO(k^2)}$. Henceforth, we assume that every arc in the hypothetical solution $S$ (if it exists) has at least one endpoint disjoint from $M$.
   Let $\sigma\subseteq (M\times M)\setminus \{(u,u)\mid u\in M\}$ be a set of ordered pairs of distinct elements from $M$ such that for every distinct $m_1,m_2\in M$, $(m_1,m_2)\in \sigma$ if and only if $m_2$ is reachable from $m_1$ in $G-S$ by a path whose internal vertices are disjoint from $M$. We call $\sigma$ the {\em reachability signature} of $M$ in $G-S$. Clearly, the number of possibilities for $\sigma$ is    $2^{\OO(k^2)}$ and so, we may assume that our algorithm has correctly guessed $\sigma$. 

Define the set $\type=\bigcup_{|M|^{2}\leq \beta \leq (|M|+k)^{2}}\{0,1\}^{\beta}$. The elements of $\type$ are called types. 
The types can be interpreted as the edge-characteristic vectors of directed graphs on at least $|M|$ and at most $|M|+k$ vertices. However, we will only apply this notion to subgraphs of $G$ that contain the vertices of $M$. 
Precisely, for  every set $X\subseteq V(G)$ containing $M$ and $Y\subseteq  E(G[X])$ the type of the subgraph $G[X,Y]$ with vertex set $X$ and arc set $Y$ is defined as its edge-characteristic vector induced by $\pi_{X}$.  Since $M$ has size bounded by $k$, the total number of types is clearly bounded by a function of $k$.

For each weakly connected component (from now onwards, simply called a component) $C$ of $G-M$, we define the graph $G_C=G[C\cup M]-E(G[M])$.  For each component $C$ of $G-M$ we compute the type of $G_C$. We also compute $n_\tau$ -- the number of components $C$ such that $G_C$ has type $\tau$, for each $\tau\in \type$. Since the type of each of these subgraphs can be computed in $f(k)$-time for some function $f$, this step takes {\fpt} time. 
 
Consider a set $X\subseteq V(G)$ containing $M$ and $Y\subseteq  E(G[X])$. We say that the subgraph $G'=G[X,Y]$ is {\em compatible} with $\sigma$ if every vertex of $X\setminus M$ is balanced in its strongly connected component of the graph $G'+\sigma$.  Here, the graph $G'+\sigma$ is obtained by starting with $G'$ and for every $(m_1,m_2)\in \sigma\setminus E(G')$, adding it to $G'$. It follows from the definition of types that for sets $X_1,X_2\subseteq V(G)$ both containing $M$ and sets $Y_1\subseteq  E(G[X_1])$ and $Y_2\subseteq  E(G[X_2])$ such that the graphs $G[X_1,Y_1]$ and $G[X_2,Y_2]$ have the same type, either both graphs are compatible with $\sigma$, or neither graph is. So, we say that a {\em type} $\tau$ is compatible with $\sigma$ if some subgraph of $G$ of that type and having $M$ in the vertex set is compatible with $\sigma$. 
For every component $C$, define $S_C=S\cap E(G_C)$ and say that the type of $G_C-S_C$ is {\em compatible} with $\sigma$ {if every vertex in $V(C)$ is balanced in its strongly connected component of the graph $G'_C=G_C-S_C+\sigma$ and {\em incompatible} with $\sigma$, otherwise}.   For each component $C$ since $V(C)$ has size bounded by $k$, the number of possibilities for $S_C$ is also bounded by a function of $k$ (here, we crucially use the fact that we have a simple digraph). Hence, in {\fpt} time, we compute a boolean table $\Gamma$ that expresses, for every $C$ and every $S_C$ that is a subset of arcs in $G_C$, whether the type of $G_C-S_C$ is compatible with $\sigma$.

    Deleting the arcs of the hypothetical solution $S$ from each $G_C$ transitions $G_C$ from a graph of one type to a graph of another type that is compatible with $\sigma$. 
    To be precise, for each $C$, we can think of the operation of deleting $S_C$ as being equivalent to the operation of moving $G_C$ from the type of $G_C$ (call it $\tau_1$) to the type of $G_C-S_C$ (call it $\tau_2$), at cost $|S_C|$.

        Moreover, the type $\tau_2$ must be compatible with $\sigma$ since $S$ is a solution and $\sigma$ has been correctly guessed. So, for every pair of types $\tau_1,\tau_2$, let us define  $\mathrm{cost}(\tau_1,\tau_2)$ to be the cost of transitioning a graph $G_C$ of type $\tau_1$ to one of type $\tau_2$.  Set this value to be prohibitively high, say, the number of arcs in $G$ plus one, if $\tau_2$ is not compatible with $\sigma$ or if there is no set $S_C$ such that $G_C-S_C$ is of type $\tau_2$.

  In our next step, for every pair of vertices $m_1,m_2\in M$ such that $(m_1,m_2)\in \sigma\setminus E(G[M])$, let us construct the set $T_{m_1,m_2}^{*}$ of all types such that for each component $C$ such that $G_C$ has type $\tau\in T_{m_1,m_2}^{*}$, there is an $m_1$-$m_2$ path with all internal vertices contained in $V(C)$. We will ensure that for every $(m_1,m_2)\in \sigma\setminus E(G[M])$, we transition into at least one of the graphs of type $\tau$ for some $\tau\in T_{m_1,m_2}^{*}$. Similarly, to ensure that we do not create an $m_1$-$m_2$ path where the reachability signature forbids one, we will ensure that for every such $(m_1,m_2)$, we {\em do not} transition into any of the graphs of type $\tau\in T_{m_1,m_2}^{*}$.

    Finally, whether or not the vertices of $M$ are balanced in $\mathrm{strong}(G-S)$ is determined entirely by the number of graphs of each type that were ``produced'' in $G-S$  by the aforementioned transitions, assuming the correct reachability signature  $\sigma$.  So, for every type compatible with $\sigma$, we determine the imbalance ``created'' by a graph of this type on each vertex of $M$. 
        To be precise, for 
a    type $\tau\in \type$ and a vertex $u\in M$, {\em the imbalance on $u$ imposed by} $\tau$ is denoted by $I(\tau,u)$ and is formally defined as follows. 
    
Consider any graph $G''$ such that $V(G'')\supseteq M$ and $M$ induces an independent set, along with an ordering $\pi'$ of $V(G')$ where the vertices of $M$ appear first  and have the same ordering relative to each other as discussed at the beginning of the proof. Moreover, suppose that the type of $G''$ implied by the ordering $\pi'$ is $\tau\in \type$. Then, an arc $(p,q)\in E(G'')$ that is incident to a vertex $u\in M$ is called  {active} if and only if $p$ and $q$ lie in the same strongly connected component of the graph  $G''+\sigma$. Now, the imbalance on a vertex $u\in M$ imposed by $\tau$ is obtained by subtracting the number of active incoming arcs of $G''$ incident to $u$ from the number of active outgoing arcs of $G''$ incident to $u$.

    \paragraph{Set up of ILP:} All of the above requirements can be formulated as an  \ilpfeas instance with $f(k)$ variables that effectively minimizes the total costs of all the required type transitions. More precisely, for every pair of types $\tau_1$ and $\tau_2$, we have a variable $x_{\tau_1,\tau_2}$ that is intended to express the number of graphs $G_C$ of type $\tau_1$ that transition to type $\tau_2$. We only need to consider variables $x_{\tau_1,\tau_2}$ where $\tau_1$ is the type of some $G_C$ and $\tau_2$ is compatible with $\sigma$. We restrict our variable set to only such variables $x_{\tau_1,\tau_2}$. Moreover, for every $\tau$ that is compatible with $\sigma$, we have a variable $y_\tau$ that is intended to express the number of components $C$ such that $G_C$ transitions to type $\tau$.

    Then, we have constraints that express the following:

    \begin{enumerate}
        \item The cost of all the type transitions is at most $p$.
            \[
                \sum_{\tau_1,\tau_2 \in \type} \mathrm{cost}(\tau_1,\tau_2) \cdot x_{\tau_1,\tau_2}\leq p
            \]

        \item 
        
For every    $(m_1,m_2)\in \sigma\setminus E(G[M])$, there is at least one transition to a type in $T_{m_1,m_2}^{*}$. 

\[
                \sum_{\tau_1\in \type}\sum_{\tau\in T_{m_1,m_2}^{*}} x_{\tau_1,\tau}\geq 1
            \]

            \item We next have to ensure that for every $(m_1,m_2)\notin \sigma$, there is no transition to a type in $T_{m_1,m_2}^{*}$ as that would imply a path between vertices of $M$ that is forbidden by $\sigma$.

\[
                \sum_{\tau_1\in \type}\sum_{\tau\in T_{m_1,m_2}^{*}} x_{\tau_1,\tau}\leq 0
            \]

        \item For every component $C$, $G_C$ transitions to some type compatible with $\sigma$. So, for every type $\tau$, we have a constraint as follows.
            \[
                \sum_{\tau_2\in \type
                } x_{\tau,\tau_2}=n_{\tau}
            \]

            Recall that $n_{\tau}$ denotes the number of components $C$ such that $G_C$ is of type $\tau$ and we have computed it already.

        \item  The number of components $C$ such that $G_C$ transitions to type $\tau$ where $\tau$ is compatible with $\sigma$, is given by summing up the values of $x_{\tau_1,\tau}$ over all possible values of $\tau_1$.
            \[
                \sum_{\tau_1\in \type} x_{\tau_1,\tau}=y_{\tau}
            \]

        \item The total imbalance imposed on each vertex of $M$ by the existing arcs incident to it, plus the imbalance imposed on it by the types to which we transition, adds up to 0.

            For each $u\in M$, let $\rho_u$ denote the imbalance on $u$ imposed by those arcs of $G[M]$ that are incident to $u$ and active in the graph $G[M]+\sigma$. The imbalance imposed on $u$ by a particular type $\tau$ is $I(\tau,u)$ and this needs to be multiplied by the number of ``occurrences'' of this type after removing the solution, i.e., the value of $y_{\tau}$.

            Hence, we have the following constraint for every $u\in M$.
            \[
                \rho_u+   \sum_{\tau\in \type} I(\tau,u) \cdot y_{\tau}=0
            \]

        \item Finally, we need all variables to get non-negative values. So, for every $\tau_1,\tau_2\in \type$,  we add $x_{\tau_1,\tau_2}\geq 0$ and for every $\tau\in \type$, $y_\tau\geq 0$.

    \end{enumerate}
    It is straightforward to convert the above constraints into the form of an instance of {\ilpfeas}. Since the number of variables is a function of $k$, Proposition~\ref{prop:lenstra} can be used to decide feasibility in {\fpt} time.
    From a solution to the {\ilpfeas} instance, it is also straightforward to recover a solution to our instance by using the table $\Gamma$.
\end{proof}

\section{Conclusions}
We have resolved the open problem of Cechlárová and Schlotter \cite{DBLP:conf/iwpec/CechlarovaS10} on the parameterized complexity of the Eulerian Strong Component Arc Deletion problem by showing that it is {\woh} and accompanied it with further hardness results parameterized by the vertex cover number and max-degree of the graph. On the positive side,  we showed that though the problem is inherently difficult in general, certain combined parameterizations (such as treewidth plus either max-degree or solution size) offer a way to obtain {\fpt} algorithms.

Our work points to several natural future directions of research on this problem.
\begin{enumerate}\item Design of ({\fpt}) approximation algorithms for \ESCAD?
\item \ESCAD parameterized by the solution size is {\fpt} on tournaments~\cite{DBLP:journals/ipl/CrowstonGJY12}. For which other graph classes is the problem {\fpt} by the same parameter?
\item Our {\fpt} algorithm for {\SESCAD} parameterized by vertex integrity is only aimed at being a characterization result and we have not attempted to optimize the parameter dependence. So, a natural follow up question is to obtain an algorithm that is as close to optimal as possible.
\item For which parameterizations upper bounding the solution size is {\ESCAD} {\fpt}? For instance, one could consider the size of the minimum directed feedback arc set of the input digraph as a parameter. Notice that in the reduction of \Cref{theo:escadsolution}, we obtain instances with unboundedly large minimum directed feedback arc sets due to the imbalance gadgets starting at the vertex $s_i$ for some color class $i$ and ending at the vertices in $\{x_u\mid u\in {\rm color~class~} i\}$.
\end{enumerate}

\noindent
\textbf{Funding:} 
The research leading to the paper was supported by UKRI EPSRC Research Grant (EP/V044621/1).

\medskip
\noindent
\textbf{Author Contributions:} All authors contributed equally at every stage of this work.

\medskip
\noindent
\textbf{Acknowledgements:} We are grateful to the anonymous reviewers for their insightful feedback, which helped us identify and correct an error in Section \ref{sec:escad_tw_tractability}. 

\bibliographystyle{plainurl}
\bibliography{main.bib}

\end{document}